\PassOptionsToPackage{dvips,final}{graphicx}

\documentclass[1p]{elsarticle}

\usepackage[english]{babel}
\usepackage[latin1]{inputenc}
\usepackage{amsmath,amssymb}
\usepackage{mathrsfs}

\newtheorem{theorem}{Theorem}[section]
\newtheorem{corollary}[theorem]{Corollary}
\newtheorem{lemma}[theorem]{Lemma}

\newtheorem{fact}[theorem]{Fact}
\newtheorem{definition}[theorem]{Definition}
\newtheorem{example}[theorem]{Example}

\newenvironment{proof}{\noindent {\bf Proof}.\ }{\ \\}
\newcommand{\spath}[1]{\mbox{[$#1$]-}path}
\newcommand{\adj}[1]{\mbox{[$#1$]-}adjacent}
\newcommand{\component}[1]{\mbox{[$#1$]-}component}
\newcommand{\components}[1]{\mbox{[$#1$]-}components}

\newcommand{\connected}[1]{\mbox{[$#1$]-}connected}
\newcommand{\touches}[1]{\mbox{[$#1$]-}touches}

\newcommand{\nodes}{\mathit{nodes}}
\newcommand{\edges}{\mathit{edges}}
\newcommand{\HD}{H\!D}
\newcommand{\HG}{{\cal H}}
\newcommand{\JT}{J\!T}

\newcommand{\node}{{\mathcal N}}
\renewcommand{\root}{\mathit{root}}
\newcommand{\vertices}{\mathit{vertices}}

\newcommand{\Pol}{\mbox{\rm P}}
\newcommand{\NP}{\mbox{\rm NP}}

\newcommand{\DB}{{\rm \mbox{\rm DB}}}

\newcommand{\vars}{\mathit{vars}}
\newcommand{\atoms}{\mathit{atoms}}

\newcommand{\V}{\mathcal{V}}

\newcommand{\C}{\mathcal{C}}
\newcommand{\E}{\mathrm{ED}}
\newcommand{\F}{\mathrm{Fr}}
\newcommand{\ecomponent}[1]{\mbox{[$#1$]-}option}
\newcommand{\ecomponents}[1]{\mbox{[$#1$]-}options}

\newcommand{\iecomponent}[1]{\mbox{\em [$#1$]-}option}

\newcommand{\tuple}[1]{\langle#1\rangle}

\newcommand{\nop}[1]{}

\newcommand{\longv}[1]{}

\journal{Information and Computation}

\begin{document}

\begin{frontmatter}

\title{Greedy Strategies and Larger Islands of Tractability for Conjunctive Queries and Constraint Satisfaction Problems}

\author[math]{Gianluigi Greco}
\ead{ggreco@mat.unical.it}
\author[deis]{Francesco Scarcello}
\ead{scarcello@dimes.unical.it}

\address[math]{Dipartimento di Matematica e Informatica, Universit\`a della Calabria, I-87036 Rende(CS), Italy}
\address[deis]{DIMES, Universit\`a della Calabria, I-87036 Rende(CS), Italy}

\begin{abstract}
Structural decomposition methods have been developed for identifying tractable classes of instances of fundamental problems in databases, such
as conjunctive queries and query containment, of the constraint satisfaction problem in artificial intelligence, or more generally of the
homomorphism problem over relational structures.
These methods work on the hypergraph structure of problem instances. Each method provides a way of transforming any cyclic hypergraph into an
acyclic one, by organizing its edges (or its nodes) into a polynomial number of clusters, and by suitably arranging these clusters as a tree,
called decomposition tree. Then, by using such a tree (or by just knowing that any exists) the given problem instance can be solved in
polynomial time.

Most structural decomposition methods can be characterized through hypergraph games that are variations of the \emph{Robber and Cops} graph
game that characterizes the notion of treewidth. In particular, decomposition trees somehow correspond to \emph{monotone winning strategies},
where the escape space of the robber on the hypergraph is shrunk monotonically by the cops. In fact, unlike the treewidth case, there are
hypergraphs where monotonic strategies do not exist, while the robber can be captured by means of more complex non-monotonic strategies.
However, these powerful strategies do not correspond in general to valid decompositions.

The paper provides a general way to exploit the power of non-monotonic strategies, by allowing a ``disciplined'' form of non-monotonicity,
characteristic of cops playing in a \emph{greedy} way. It is shown that deciding the existence of a (non-monotone) greedy winning strategy (and
compute one, if any) is tractable. Moreover, despite their non-monotonicity, such strategies always induce valid decomposition trees, which can
be computed efficiently based on them.
%
As a consequence, greedy strategies allow us to define new islands of tractability for the considered problems, properly including all
previously known classes of tractable instances. In particular, we define the new notion of greedy hypertree decomposition of a hypergraph,
whose associated notion of width is at most the hypertree width, and sometimes strictly smaller.
\end{abstract}

\begin{keyword}
Structural Decomposition Methods \sep Games on Discrete Structures \sep Conjunctive Queries and Databases \sep Constraint Satisfaction Problems \sep Hypertree Decompositions \sep Tree Projections \sep Homomorphism Problem
\end{keyword}
\end{frontmatter}

\newpage

\section{Introduction}

We look for islands of tractability for answering conjunctive queries over relational databases or, equivalently, for solving constraint
satisfaction problems. For the sake of presentation, we next focus on the database setting and the conjunctive query answering problem. We
remark that all results can immediately be applied to all problems that can be recast as homomorphism  problems,  and possibly can be useful in
further settings, thanks to the general combinatorial nature of the proposed approach. We refer the interested reader to~\cite{GGSBook} for
more detail on the connections with the homomorphism problem and with further equivalent problems.

\subsection{Acyclic Conjunctive Queries}\label{sec:acyclic}

Conjunctive queries are defined through conjunctions of atoms (without negation), and are known to be equivalent to Select-Project-Join
queries. The problem of evaluating such queries is $\NP$-hard in general, but it is feasible in polynomial time on the class of acyclic queries
(we omit ``conjunctive,'' hereafter), which was the subject of many seminal research works since the early ages of database theory (see,
e.g.,~\cite{BFMY83}). This class contains all queries $Q$ whose associated query hypergraph $\HG_Q$ is acyclic,\footnote{For completeness,
observe that different notions of hypergraph acyclicity have been proposed in the literature. This paper follows the standard definition of
acyclic conjunctive queries, so that hypergraph acyclicity always refers to the most liberal notion, known as
$\alpha$-acyclicity~\cite{fagi-83}.} where $\HG_Q$ is a hypergraph having the variables of $Q$ as its nodes, and the (sets of variables
occurring in the) atoms of $Q$ as its hyperedges.
In fact, queries arising from real applications are hardly precisely acyclic. Yet, they are often not very intricate and, in fact, tend to
exhibit some limited degree of cyclicity, which suffices to retain most of the nice properties of acyclic ones.
Therefore, several efforts have been spent to investigate invariants that are best suited to identify nearly-acyclic hypergraphs, leading to
the definition of a number of so-called {\em (purely) structural decomposition-methods}, such as the \emph{(generalized)
hypertree}~\cite{gott-etal-99}, \emph{fractional hypertree}~\cite{GM06}, \emph{spread-cut}~\cite{CJG08}, and \emph{component
hypertree}~\cite{GMS07} decompositions. These methods aim at transforming a given cyclic hypergraph into an acyclic one, by organizing its
edges (or its nodes) into a polynomial number of clusters, and by suitably arranging these clusters as a tree, called decomposition tree. The
original problem instance can then be evaluated over such a tree of subproblems, with a cost that is exponential in the cardinality of the
largest cluster, also called {\em width} of the decomposition, and polynomial if this width is bounded by some constant.

Despite their different technical definitions,  there is a simple mathematical framework, based on the notion of \emph{tree
projection}~\cite{GS84}, that encompasses all the above decomposition methods, as pointed out in recent works on the subject~\cite{GS08,GS10b}.
In this setting, a query $Q$ is given together with a set $\V$ of atoms, called views, which are defined over the variables in $Q$.
The question is whether (parts of) the views can be arranged as to form a tree projection (playing the role of a decomposition tree), i.e., a
novel acyclic query that still ``covers'' $Q$.
By representing $Q$ and $\V$ via the hypergraphs $\HG_Q$ and $\HG_\V$, where hyperedges one-to-one correspond with query atoms and views,
respectively, the tree projection problem reveals its graph-theoretic nature. For a pair of hypergraphs $\HG_1,\HG_2$, let $\HG_1\leq \HG_2$
denote that each hyperedge of $\HG_1$ is contained in some hyperedge of $\HG_2$. Then, a tree projection of $\HG_Q$ w.r.t.~$\HG_\V$ is any
acyclic hypergraph $\HG_a$ such that $\HG_Q\leq \HG_a\leq \HG_\V$. If such a hypergraph exists, then we say that the pair of hypergraphs
$(\HG_Q,\HG_\V)$ has a tree projection.\footnote{Note that the only known decomposition technique that does not fit the above framework is the
one based on the \emph{submodular width}~\cite{M10}. This method is in fact not ``purely'' structural, in that the views $\V$, together with
suitable associated database relations, are computed in fixed-parameter polynomial time (hence, not in polynomial-time, in general) by using
the actual database over which $Q$ has to be evaluated, rather than looking at $\HG_Q$ only.}

\begin{figure}[t]
\centering
\includegraphics[width=0.99\textwidth]{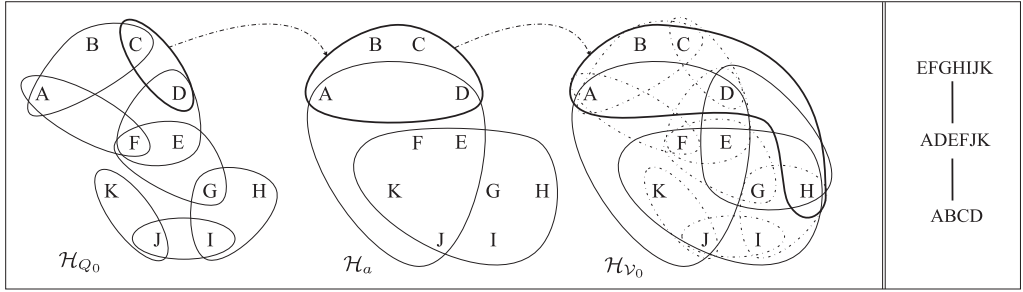}
\caption{A tree projection $\HG_a$ of $\HG_{Q_0}$ w.r.t.~$\HG_{\V_0}$; 
On the right: A join tree $\JT_a$ for $\HG_a$.
}\label{fig:hypergraph}
\end{figure}

\begin{example}\label{ex:intro}\em Consider the conjunctive query
\vspace{-1mm}
\[
\begin{array}{ll}
Q_0:  & r_1(A,B,C)\wedge r_2(A,F)\wedge r_3(C,D)\wedge r_4(D,E,F)\wedge\\
& r_5(E,F,G) \wedge r_6(G,H,I)\wedge r_7(I,J)\wedge r_8(J,K),\\
\end{array}
\]

\vspace{-1mm}\noindent whose associated hypergraph $\HG_{Q_0}$ is depicted in Figure~\ref{fig:hypergraph}, together with other hypergraphs that
are discussed next.

To answer $Q_0$, assume that a set $\V_0$ of views is available comprising some views, called {\em query views}, playing the role of query
atoms, plus four additional views. The set of variables of each view is a hyperedge in the hypergraph $\HG_{\V_0}$ (query views are depicted as
dashed hyperedges).
In the middle between $\HG_{Q_0}$ and $\HG_{\V_0}$, Figure~\ref{fig:hypergraph} reports the hypergraph $\HG_a$ which covers $\HG_{Q_0}$, and
which is in its turn covered by $\HG_{\V_0}$---e.g., $\{C,D\}\subseteq \{A,B,C,D\}\subseteq\{A,B,C,D,H\}$. Since $\HG_a$ is in addition
acyclic, $\HG_a$ is a tree projection of $\HG_{Q_0}$ w.r.t.~$\HG_{\V_0}$.~\hfill~$\lhd$
\end{example}

Observe that, in the tree projection framework, views can be arbitrary, i.e, they do not depend on the specific conjunctive query $Q$, and can
be reused to answer different queries. In particular, views may be the materialized output of any procedure over the database, possibly much
more powerful than conjunctive queries.
Moreover, it is known and easy to see that any decomposition method based on clustering subproblems can be viewed as an instance of this
general setting, identifying  a specific set of views to answer a given query $Q$ efficiently (see Section~\ref{sec:framework}).

\subsection{Islands of Tractability}

An \emph{island of tractability} (cf.~\cite{K03}) in the tree projection  framework is a class $\C$ of pairs $(Q,\V)$ that can be efficiently
recognized, i.e., we can check in polynomial time whether a given pair actually belongs to $\C$, and such that $Q$ can be efficiently evaluated
on every database, by possibly exploiting the views that are available in~$\V$.

Many specializations of tree projections, such as \emph{tree decompositions}~\cite{RS84}, \emph{hypertree
decompositions}~\cite{gott-etal-99},\emph{component decompositions}~\cite{GMS07}, and \emph{spread-cuts decompositions}~\cite{CJG08}, define
islands of tractability whenever some fixed bound is imposed on their widths. This is also the case for \emph{fractional hypertree
decompositions}~\cite{GM06}, whenever the resources sufficient for computing their $O(w^3)$ approximation~\cite{M09} are used as available
views.
However, this is not the case for general tree projections. Indeed, while Goodman and Shmueli~\cite{GS84} observed that queries that admit a
tree projection can be evaluated in polynomial time, Gottlob et al.~\cite{GMS07} proved that checking whether a tree projection exists or not
is an $\NP$-hard problem. Hence, the class $\C_{tp}=\{(Q,\V) \mid \HG_Q \mbox{ has a tree projection w.r.t.~} \HG_\V  \}$, which includes all
the above mentioned islands of tractability, is not an island of tractability in its turn.
%
%
%
A natural question is, therefore, whether there is any subclass of $\C_{tp}$, at least including all the tractable classes mentioned above,
which identifies an actual island of tractability where tree projections can be computed efficiently.

In the paper, we address the above question.
The starting point of our analysis is the game-theoretic characterization of tree projections in terms of the \emph{Robber and Captain}
game~\cite{GS08}. The game is played on a pair of hypergraphs $(\HG_1,\HG_2)$ by a Captain controlling, at each move, a squads of cops encoded
as the nodes in a hyperedge $h\in \edges(\HG_2)$, and by a Robber who stands on a node and can run at great speed along the edges of $\HG_1$,
while being not permitted to run trough a node that is controlled by a cop. In particular, the Captain may ask any cop in the squad $h$ to run
in action, as long as they occupy nodes that are currently reachable by the Robber, thereby blocking an escape path for the Robber. While cops
move, the Robber may run trough those positions that are left by cops or not yet occupied. The goal of the Captain is to place a cop on the
node occupied by the Robber, while the Robber tries to avoid her capture. The Captain has a winning strategy if, and only if, there is a tree
projection of $\HG_1$ w.r.t.~$\HG_2$.

Based on the above characterization, we proceed as follows:

\begin{itemize}
\item[$\blacktriangleright$] We define the notion of {\em greedy strategies}, which are winning strategies for the Captain, possibly
    non-monotone, where it is required that \emph{all} cops available at the current squad $h$ and reachable by the Robber enter in action.
    If all of them are in action, then a new squad $h'$ is selected, again requiring that all the active cops, i.e., those in the frontier,
    enter in action. In the Robber and Captain game, it is known that there is no incentive for the Captain to play a strategy that is not
    monotone~\cite{GS08}. Instead, by focusing on greedy strategies, we can exhibit examples where there exists non-monotone winning
    strategies but no monotone winning one.

\item[$\blacktriangleright$] We show that {\em greedy strategies} can be computed in polynomial time, and that based on them (even on
    non-monotone ones) it is possible to construct, again in polynomial time, tree projections, which are called \emph{greedy}. Therefore,
    the class $\C_{gtp}\subset \C_{tp}$ of all {greedy tree projections} turns out to be an island of tractability.

\item[$\blacktriangleright$] We show that $\C_{gtp}$ properly includes most previously known islands of tractability (based on structural
    properties), precisely because of the power of non-monotonic strategies. Indeed, (arbitrary) non-monotone strategies do not correspond
    in general to valid decompositions in the games characterizing such islands of tractability, which are in fact defined in terms of
    monotone strategies only. The novel notion of greedy tree projections allows us to define new islands of tractability from any known
    structural decomposition method. In particular, from the notion of generalized hypertree decomposition, we obtain the novel notion of
    {\em greedy (generalized) hypertree decomposition}, that is tractable and strictly more powerful than the hypertree decomposition
    (which is instead characterized by a monotonic hypergraph game).

\item[$\blacktriangleright$] Finally, by using the game theoretic characterization of tree projections, we pinpoint that dealing with this
    general $\NP$-hard notion is fixed-parameter tractable if the maximum arity of views is used as the parameter. Even this result can
    be useful in real-world applications, since the case of small arity structures is quite frequent in practice.
\end{itemize}

\noindent \textbf{Organization.}
The rest of paper is organized as follows.
Section~\ref{sec:framework} illustrates some basic notions and concepts. Greedy strategies for the Robber and Captain game are introduced and
analyzed in Section~\ref{SECGAMES}, and based on them islands of tractability for tree projections are singled out in Section~\ref{sec:larger}.
Specializations of the results to known structural decomposition methods (as well as to structures having ``small'' arities) are discussed in
Section~\ref{sec:methods}.
Literature related to ``Cops and Robbers'' games is illustrated in Section~\ref{sec:related}, while a few remarks and open issues are
discussed in Section~\ref{sec:conclusion}.

\section{Preliminaries}\label{sec:framework}

\noindent \textbf{Hypergraphs and Acyclicity.} A \emph{hypergraph} $\HG$ is a pair $(V,H)$, where $V$ is a finite set of nodes and $H$ is a set
of hyperedges such that, for each $h\in H$, $h\subseteq V$.
If $|h|=2$ for each (hyper)edge $h\in H$, then $\HG$ is a {\em graph}.
We assume without loss of generality that every node occurs in some hyperedge, that is, $V= \bigcup_{h\in H} h$.
We denote $V$ and $H$ by $\nodes(\HG)$ and $\edges(\HG)$, respectively.

A hypergraph $\HG$ is {\em acyclic} (more precisely, $\alpha$-acyclic~\cite{fagi-83}) if, and only if, it has a join tree~\cite{bern-good-81}.
A {\em join tree} $\JT$ for a hypergraph $\HG$ is a tree whose vertices are the hyperedges of $\HG$ such that, whenever a node $X\in V$ occurs
in two hyperedges $h_1$ and $h_2$ of $\HG$, then $h_1$ and $h_2$ are connected in $\JT$, and $X$ occurs in each vertex on the unique path
linking $h_1$ and $h_2$. In words, the set of vertices in which $X$ occurs induces a (connected) subtree of $\JT$. We will refer to this
condition as the {\em connectedness condition} of join trees.

\begin{example}\label{ex:sec2}\em
Consider the hypergraph $\HG_a$ reported in Figure~\ref{fig:hypergraph}. We have $\nodes(\HG_a)=\{A,B,C,D,E,F,G,H,I,J,K\}$ and
$\edges(\HG_a)=\{\{A,B,C,D\},$ $\{A,D,E,F,J,K\},$ $\{E,$ $F,G,H,I,J,K\}\}$. The hypergraph is acyclic, as it is witnessed by the join tree
$\JT_a$ depicted on the right part of the same figure. \hfill $\lhd$
\end{example}

\medskip \noindent \textbf{Tree Decompositions of (Hyper)graphs.}
Several efforts have been spent in the literature to investigate hypergraph properties that are best suited to identify nearly-acyclic
hypergraphs, leading to the definition of a number of so-called {\em (purely) structural decomposition methods}.
Within these methods, the notions of tree decomposition and treewidth~\cite{RS84} represent a significant success story in Computer Science
(see, e.g., \cite{GGSBook}), which are meant to provide a measure of the degree of cyclicity in \emph{graphs}.

A \emph{tree decomposition}~\cite{RS84} of a graph $G$ is a pair $\tuple{T,\chi}$, where $T=(N,E)$ is a tree, and $\chi$ is a labeling
function assigning to each vertex $v\in N$ a set of vertices $\chi(v)\subseteq \nodes(G)$, such that the following conditions are satisfied:
(1) for each node $Y\in \nodes(G)$, there exists $p\in N$ such that $Y\in \chi(p)$; (2) for each edge  $\{X,Y\}\in \edges(G)$, there exists
$p\in N$ such that $\{X,Y\}\subseteq \chi(p)$; and (3) for each node $Y\in \nodes(G)$, the set $\{p\in N \mid  Y\in \chi(p)\}$ induces a
(connected) subtree of $T$. The \emph{width} of $\tuple{T,\chi}$ is the number $\max_{p\in N}(|\chi(p)|-1)$.

For the application of the notion of treewidth over an arbitrary hypergraph, it is necessary to deal with a graph-based representation of its
associated hypergraph. There are a number of possible choices, and we next focus on the simplest and widely used one.

The \emph{Gaifman graph} of a hypergraph $\HG$ is defined over the set $\nodes(\HG)$ of the nodes of $\HG$, and contains an edge $\{X,Y\}$ if,
and only if, $\{X,Y\}\subseteq h$ holds, for some hyperedge $h\in\edges(\HG)$. The {\em treewidth} of $\HG$ is the minimum width over all the
tree decompositions of its Gaifman graph. Deciding whether a given hypergraph has treewidth bounded by a fixed natural number $k$ is known to
be feasible in linear time~\cite{bodl-96}.

\begin{figure}[t]
 \centering
 \includegraphics[width=0.9\textwidth]{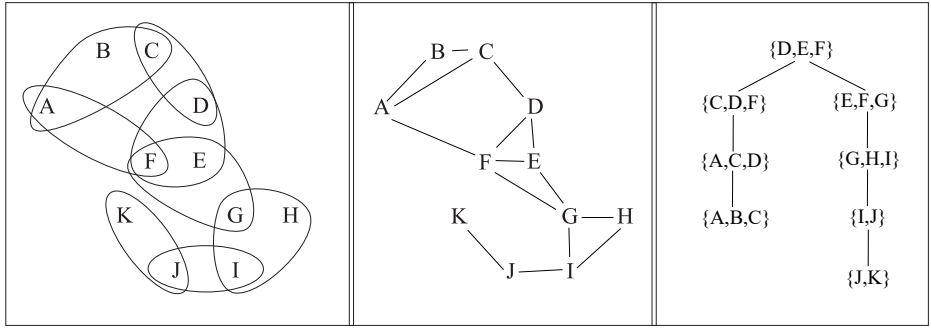}
 \caption{The hypergraph $\HG_{Q_0}$, its Gaifman graph, and a tree decomposition of it.}\label{fig:sec2}
\end{figure}

\begin{example}\em
Consider the hypergraph $\HG_{Q_0}$ discussed in Example~\ref{ex:intro} and reported again in Figure~\ref{fig:sec2}, for the sake of
readability. The hypergraph $\HG_{Q_0}$ is not acyclic, as it is not possible to build a join tree for it.
In fact,  Figure~\ref{fig:sec2} also reports the Gaifman graph of $\HG_{Q_0}$ and a tree decomposition of it. Note that there are vertices of
the tree decomposition containing 4 nodes of $\HG_{Q_0}$. Indeed, the treewidth of $\HG_{Q_0}$ is 3. \hfill $\lhd$
\end{example}

\medskip \noindent \textbf{(Generalized) Hypertree Decompositions of Hypergraphs.}
A crucial limitation for the practical use of the tree decomposition method is that it applies to graph representations only, hence obscuring
in many cases the actual degree of cyclicity of the original hypergraph.
For instance, for the acyclic hypergraph $\HG_a$ depicted in Figure~\ref{fig:hypergraph}, the Gaifman graph contains a clique over the
variables in $\{A,D,E,F,J,K\}$, since all of them occur together in one hyperedge. Hence, the treewidth of this acyclic hypergraph is 5.
Motivated by this observation, specific width-notions for hypergraphs have been defined and studied, and often these are more effective than
simply applying the treewidth on a suitable ``binarization''~\cite{GS10c}. In particular, the natural counterpart of the tree decomposition
method over hypergraphs is the notion of (generalized) hypertree decomposition~\cite{gott-etal-03} (see~\cite{GGLS16} for a survey on recent
advances and applications).

A {\em hypertree for a hypergraph $\HG$} is a triple $\tuple{T,\chi,\lambda}$, where $T=(N,E)$ is a rooted tree, and $\chi$ and $\lambda$ are
labeling functions which associate each vertex $p\in N$ with two sets $\chi(p)\subseteq \nodes(\HG)$ and $\lambda(p)\subseteq \edges(\HG)$. If
$T'=(N',E')$ is a subtree of $T$, we define $\chi(T')= \bigcup_{v\in N'} \chi(v)$.
In the following, for any rooted tree $T$, we denote the set of vertices $N$ of $T$ by $\vertices(T)$, and the root of $T$ by $\root(T)$.
Moreover, for any $p\in N$, $T_p$ denotes the subtree of $T$ rooted at $p$.

A {\em generalized hypertree decomposition}~\cite{gott-etal-03} of a hypergraph $\HG$ is a hypertree $\HD=\tuple{T,\chi,\lambda}$ for $\HG$
such that: (1) for each hyperedge $h\in \edges(\HG)$, there exists $p\in \vertices(T)$ such that $h\subseteq \chi(p)$; (2) for each node $Y\in
\nodes(\HG)$, the set $\{ p\in \vertices(T) \mid  Y \in \chi(p) \}$ induces a (connected) subtree of $T$; and (3) for each $p\in \vertices(T)$,
$\chi(p)\subseteq \nodes(\lambda(p))$.
The {\em width} of a generalized hypertree decomposition $\tuple{T,\chi,\lambda}$ is $max_{p\in \vertices(T)} |\lambda(p)|$. The {\em
generalized hypertree width} $ghw(\HG)$ of $\HG$ is the minimum width over all its
generalized hypertree decompositions. 
The notions is a true generalizations of acyclicity, as the acyclic hypergraphs are precisely those hypergraphs having generalized hypertree
width one.

Note that conditions (1) and (2) above state that $\tuple{T,\chi}$ is a tree decomposition of the Gaifman graph of $\HG$, while condition (3)
prescribes that, at each vertex $p$, all nodes in the $\chi$ labeling are covered by hyperedges in the $\lambda$ labeling. Indeed, the width of
the generalized hypertree decomposition is defined in terms of the number of hyperedges used to cover the nodes, rather than of the number of
such nodes, as in the width of $\tuple{T,\chi}$.

\begin{figure}[t]
 \centering
 \includegraphics[width=0.6\textwidth]{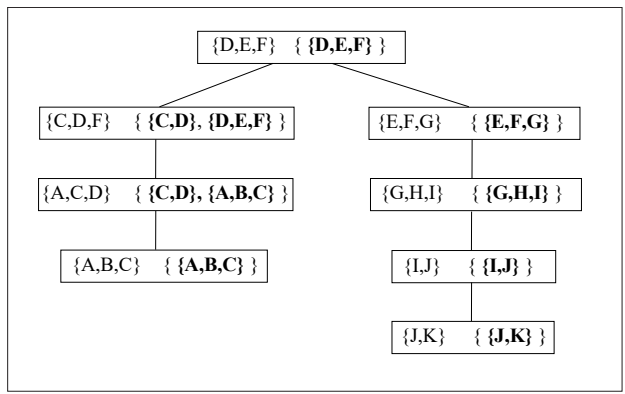}\vspace{-2mm}
 \caption{A (generalized) hypertree decomposition of the hypergraph $\HG_{Q_0}$.}\label{fig:sec2-bis}
\end{figure}

\begin{example}\em
Consider again the hypergraph $\HG_{Q_0}$ and the tree decomposition reported in Figure~\ref{fig:sec2}. The nodes occurring at each vertex of
the decomposition can be covered by using two hyperedges at most, as illustrated in the generalized hypertree decomposition depicted in
Figure~\ref{fig:sec2-bis}. Therefore, the width of this decomposition is $2$ and thus $ghw(\HG_{Q_0})\leq 2$.
Actually, $ghw(\HG_{Q_0})=2$ because $\HG_{Q_0}$ is a cyclic hypergraph, which entails $ghw(\HG_{Q_0})>1$. \hfill
$\lhd$
\end{example}

A \emph{hypertree decomposition}~\cite{gott-etal-99} of $\HG$ is a generalized hypertree decomposition $\HD=\tuple{T,\chi,\lambda}$ where: (4)
for each $p\in \vertices(T)$, $\nodes(\lambda(p)) \cap \chi(T_p) \;\subseteq\; \chi(p)$. Note that the inclusion in the above condition is
actually an equality, because Condition~(3) implies the reverse inclusion. The {\em hypertree width} $hw(\HG)$ of $\HG$ is the minimum width
over all its hypertree decompositions.
Let $k$ be any fixed natural number.
For any hypergraph $\HG$, deciding whether $hw(\HG)\leq k$ is feasible in polynomial time (and, actually, it is
highly-parallelizable)~\cite{gott-etal-99}, while deciding whether $ghw(\HG)\leq k$ is $\NP$-complete~\cite{GMS07}.

Therefore, condition (4) plays the technical role of guaranteeing that the hypertree decomposition is a tractable structural method.
Moreover, it cannot be much larger than its generalized variant, since $ghw(\HG)\leq hw(\HG)\leq 3\times ghw(\HG)+1$ holds~\cite{AGG07}.
%
As an example, the reader might check that $hw(\HG_{Q_0})=2$, too, because the generalized hypertree decomposition depicted in
Figure~\ref{fig:sec2-bis} satisfies condition~(4), and thus it is actually a hypertree decomposition. Later in the paper, Figure~\ref{fig:GMS}
shows a hypergraph where the generalized hypertree width is strictly smaller than the hypertree width.

\medskip \noindent\textbf{Tree Projections.} The framework of the {tree projections} is a mathematical framework that
encompasses all (purely) structural decomposition methods defined in the literature. Formally,
for two hypergraphs $\HG_1$ and $\HG_2$, we write $\HG_1\leq \HG_2$ if, and only if, each hyperedge of $\HG_1$ is contained in at least one
hyperedge of $\HG_2$. Let $\HG_1\leq \HG_2$; then, a \emph{tree projection} of $\HG_1$ with respect to $\HG_2$ is an acyclic hypergraph $\HG_a$
such that $\HG_1\leq \HG_a \leq \HG_2$. If such a hypergraph $\HG_a$ exists, then we say that the pair $(\HG_1,\HG_2)$ has a tree projection.
See Figure~\ref{fig:hypergraph} for an example.
Without loss of generality, we assume hereafter that $\HG_1$ and $\HG_2$ have the same set of nodes. Indeed, $\nodes(\HG_1) \not\subseteq
\nodes(\HG_2)$ trivially entails that there are no tree projections of $\HG_1$ with respect to $\HG_2$, while  $\nodes(\HG_2) \not\subseteq
\nodes(\HG_1)$ entails that there are useless nodes in $\HG_2$.

According to this unifying view, differences among the various (purely) structural decomposition methods just come in the way the resource
hypergraph $\HG_2$ is defined.
For instance, given a hypergraph $\HG$ and a natural number $k>0$, let $\HG^k$ denote the hypergraph over the same set of nodes as $\HG$, and
whose set of hyperedges is given by all possible unions of $k$ edges in $\HG$, i.e., $\edges(\HG^k)= \{ h_1 \cup h_2 \cup \cdots \cup h_k \mid
\{h_1,h_2,\ldots,h_k\}\subseteq \edges(\HG)\}$. Then, it is well known and easy to see that $ghw(\HG)\leq k$ if, and only if, there is a tree
projection for $(\HG,\HG^k)$.

Similarly, let $\HG^{tk}$ be the hypergraph over the same set of nodes as $\HG$, and whose set of hyperedges is given by all possible clusters
$B\subseteq\nodes(\HG)$ of nodes such that $|B| \leq k+1$. Then, $\HG$ has treewidth at most $k$ if, and only if, there is a tree projection
for $(\HG,\HG^{tk})$.

However, the notion of tree projection is more general then both treewidth and hypertree width, because the hyperedges of the resource
hypergraph may model arbitrary subproblems of the given instance whose solutions are easy to compute, or already available from previous
computations. For instance, in Example~\ref{ex:intro}, the resource hypergraph $\HG_{\V_0}$ does not correspond to any of the above mentioned
decomposition methods.

\medskip\noindent\textbf{Conjunctive Queries.} We leave the section by recalling conjunctive queries and their hypergraph-based representation,
 over which structural decomposition methods can be applied---as introduced in Example~\ref{ex:intro}.

A {\em conjunctive query} $Q$ consists of a finite conjunction of atoms of the form $r_1({\bf u_1})\wedge\cdots\wedge r_m({\bf u_m})$, where
$r_1,...,r_m$ (with $m > 0$) are relation symbols (not necessarily distinct), and ${\bf u_1},...,{\bf u_m}$ are lists of terms (i.e., variables
or constants).
The set of all atoms in $Q$ is denoted by $\atoms(Q)$. For a set of atoms $A$, $\vars(A)$ is the set of variables occurring in the atoms in
$A$. For short, $\vars(Q)$ denotes $\vars(\atoms(Q))$.

There is a very natural way to associate a hypergraph $\HG_\V=(N,H)$ with any set $\V$ of atoms: the set $N$ of nodes consists of all variables
occurring in $\V$; for each atom in $\V$, the set $H$ of hyperedges contains a hyperedge including all its variables; and no other hyperedge is
in $H$.
For a query $Q$, the hypergraph associated with $\atoms(Q)$ is briefly denoted by $\HG_Q$. If $\HG_Q$ is a connected hypergraph, we say that
$Q$ is a {\em connected} query.

\section{Greedy Strategies in Robber and Captain Games}\label{SECGAMES}

In this section, we define the concept of greedy strategies in the game-theoretic characterization of tree projections proposed in~\cite{GS08},
and we show that, unlike arbitrary strategies, greedy ones can be efficiently computed.

To formalize our results, we need to introduce some additional definitions and notations, which will be intensively used in the following.

Assume that a hypergraph $\HG$ is given. Let $V$, $W$, and $\{X,Y\}$ be sets of nodes. Then, $X$ is said \adj{V}\ (in $\HG$) to $Y$ if there
exists a hyperedge $h\in \edges(\HG)$ such that $\{X,Y\}\subseteq (h -V)$. A \spath{V}\ from $X$ to $Y$ is a sequence $X=X_0,\ldots,X_\ell=Y$
of nodes such that $X_{i}$ is \adj{V}\ to $X_{i+1}$, for each $i\in [0...\ell\mbox{-}1]$. We say that $X$ \touches{V} $Y$ if $X$ is
\adj{\emptyset} to $Z\in \nodes(\HG)$, and there is a \spath{V} from $Z$ to $Y$; similarly, $X$ \touches{V} the set $W$ if $X$ \touches{V} some
node $Y\in W$.
We say that $W$ is \connected{V}\ if $\forall X,Y\in W$ there is a \spath{V}\ from $X$ to $Y$. A \component{V} (of $\HG$) is a maximal
\connected{V}\ non-empty set of nodes $W\subseteq (\nodes(\HG)-V)$. For any \component{V}\ $C$, let $\edges(C) = \{ h\in \edges(\HG)\;|\; h\cap
C\neq\emptyset\}$, and for a set of hyperedges $H\subseteq \edges(\HG)$, let $\nodes(H)$ denote the set of nodes occurring in $H$, that is
$\nodes(H)=\bigcup_{h\in H} h$. For any component $C$ of $\HG$, we denote by $\F(C,\HG)$  the \emph{frontier} of $C$ (in $\HG$), i.e., the set
$\nodes(\edges(C))$.\footnote{The choice of the term ``frontier'' to name the union of a component with its outer border is due to the role
that this notion plays in the hypergraph game described in the subsequent section.}
Moreover, $\partial (C,\HG)$ denote the {\em border} of $C$ (in $\HG$), i.e., the set $\F(C,\HG)\setminus C$. Note that $C_1\subseteq C_2$
entails $\F(C_1,\HG)\subseteq \F(C_2,\HG)$.

In the following sections, given any pair of hypergraphs $(\HG_1,\HG_2)$ and a set of nodes $C\subseteq \HG_1$, we write for short $\F(C)$ and
$\partial C$ to denote $\F(C,\HG_1)$ and  $\partial (C,\HG_1)$, respectively.

\subsection{Game-Theoretic Characterization}\label{GTC}

The \emph{Robber and Captain} game is played on a pair of hypergraphs $(\HG_1,\HG_2)$ by a Robber and a Captain controlling some squads of
cops, in charge of the surveillance of a number of strategic targets. The Robber stands on a node and can run at great speed along the edges of
$\HG_1$. However, (s)he is not permitted to run trough a node that is controlled by a cop. Each move of the Captain involves one squad of cops,
which is encoded as a hyperedge $h\in \edges(\HG_2)$. The Captain may ask some cops in the squad $h$ to run in action, as long as they occupy
nodes that are currently reachable by the Robber, thereby blocking an escape path for the Robber. Thus, ``second-lines'' cops cannot be
activated by the Captain. Note that the Robber is fast and may see cops that are entering in action. Therefore, while cops move, the Robber may
run trough those positions that are left by cops or not yet occupied. The goal of the Captain is to place a cop on the node occupied by the
Robber, while the Robber tries to avoid her/his capture.

\begin{definition}\em  Let $\HG_1$ and $\HG_2$ be two hypergraphs.
The Robber and Captain game on $(\HG_1,\HG_2)$ is formalized as follows.
A \emph{position} for the Captain is a pair $(h,M)$ where $h$ is a hyperedge of $\HG_2$ and $M\subseteq h$.
A \emph{configuration} is a triple $(h,M,C)$, where $(h,M)$ is a position for the Captain, and $C$ is the \component{M} where the Robber
stands.\footnote{It is easy to see that in such games, being the robber arbitrarily fast, what matters is not the precise node where the robber
stands, but just the \component{M} where (s)he is free to move.} The initial configuration is $(\emptyset,\emptyset,\nodes(\HG_1))$.

A \emph{strategy} $\sigma$ is a function that encodes the \emph{moves} of the Captain. Its domain includes the initial configuration. For each
configuration $v_p=(h_p,M_p,C_p)$ in the domain of~$\sigma$, $\sigma(v_p)=(h_{r},M_{r})$, with $M_{r}\subseteq h_r\cap\F(C_p)$, is the novel
position for the Captain. After this move, the Robber can select any {\em \iecomponent{v_p,M_{r}}}, i.e., any \component{M_{r}} $C_{r}$ such
that $C_{p}\cup C_{r}$ is \connected{M_{p}\cap M_{r}}. If there is no \ecomponent{v_p,M_{r}}, then $(h_{r},M_{r},\emptyset)$ is said a
\emph{capture} configuration induced by $\sigma$. The move of the Captain is \emph{monotone} if, for each \ecomponent{v_p,M_{r}} $C_{r}$,
$C_{r}\subseteq C_p$. The domain of $\sigma$ includes the configuration $(h_{r},M_{r},C_{r})$, for each \ecomponent{v_p,M_{r}} $C_{r}$. No
other configuration is in the domain of $\sigma$.
The strategy $\sigma$ is monotone if it encodes only monotone moves over the configurations in its domain.

A strategy $\sigma$ can be represented as a directed graph $G(\sigma)=(N,A)$, called {\em strategy graph}, as follows. The set $N$ of nodes is
the set of all configurations in the domain of $\sigma$ plus all capture configurations induced by $\sigma$. If $v_p=(h_p,M_p,C_p)$ is a
configuration and $\sigma(v_p)=(h_{r},M_{r})$, then $A$ contains an arc from $v_p$ to $(h_{r},M_{r},C_{r})$ for each \ecomponent{v_p,M_{r}}
$C_{r}$, and to $(h_{r},M_{r},\emptyset)$ if there is no \ecomponent{v_p,M_{r}}.
We say that $\sigma$ is a winning strategy (for the Captain) if $G(\sigma)$ is acyclic. Otherwise, i.e., if $G(\sigma)$ contains a cycle, then
the Robber can avoid her/his capture forever. 
\end{definition}

\begin{figure}[t]
\centering
\includegraphics[width=0.98\textwidth]{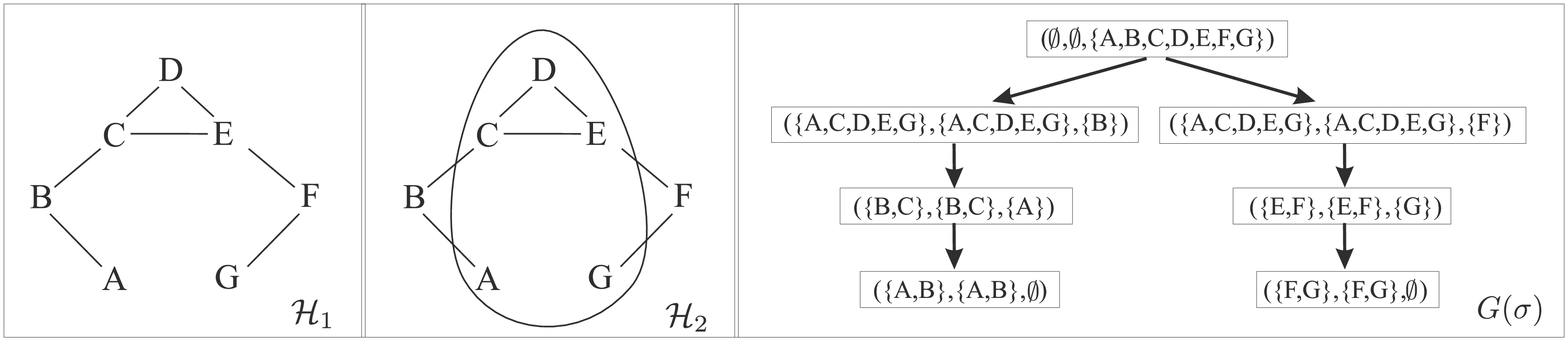}
\caption{The hypergraphs $\HG_1$ and $\HG_2$, plus the graph $G(\sigma)$ in Example~\ref{ex:construction}.}\label{fig:greedy-bis}
\end{figure}

\begin{example}\label{ex:construction}\em
Consider the two hypergraphs $\HG_1$ and $\HG_2$ reported in Figure~\ref{fig:greedy-bis}, together with the strategy graph $G(\sigma)$. The
graph encodes a winning strategy $\sigma$ for the Captain. From the initial configuration $(\emptyset,\emptyset,\nodes(\HG_1))$, the Captain
activates all the cops in the hyperedge $\{A,C,D,E,G\}$, so that the Robber has two available options, i.e., $\{B\}$ and $\{F\}$. In the former
(resp., latter) case, the Captain activates all the cops in the hyperedge $\{B,C\}$ (resp., $\{E,F\}$), so that the Robber has necessarily to
occupy the node $A$ (resp., $G$). Finally, the Captain activates the cops in $\{A,B\}$ (resp., $\{F,G\}$) and captures the Robber.
Note that the strategy $\sigma$ is non-monotone, because the Robber is allowed to return on $A$ and $G$, after that these nodes have been
previously occupied by the Captain in the first move. \hfill $\lhd$
\end{example}

In the above example, the hyperedge $\{A,C,D,E,G\}$ of $\HG_2$ ``absorbs'' the cycle in $\HG_1$, so that it is easily seen that there is a tree
projection $\HG_a$ of $\HG_1$ w.r.t.~$\HG_2$ (see Figure~\ref{fig:strategy-new-bis}). The fact that on this pair the Captain has a winning
strategy is not by chance.

\begin{figure}[t]
\centering
\includegraphics[width=0.99\textwidth]{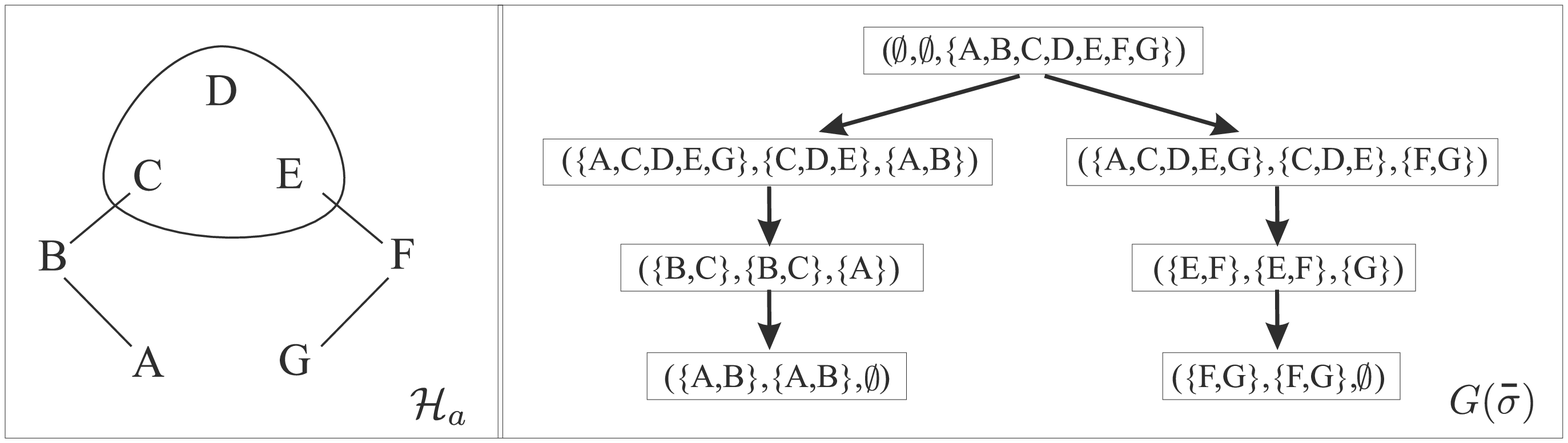}
\caption{A tree projection $\HG_a$ for the pair in Example~\ref{ex:construction}, plus the graph $G(\bar \sigma)$.}\label{fig:strategy-new-bis}
\end{figure}

\begin{theorem}[\cite{GS08}]
There is a tree projection of $\HG_1$ w.r.t.~$\HG_2$ if, and and only if, there is a winning strategy in the Robber and Captain game played on
$(\HG_1,\HG_2)$.
\end{theorem}

Recall that the winning strategy in Example~\ref{ex:construction} is not monotone. However, an important property of this game is that there is
no incentive for the Captain to play a strategy that is not monotone.

\begin{theorem}[cf. \cite{GS08}]\label{thm:costruzione}
In the Robber and Captain game played on the pair $(\HG_1,\HG_2)$, a winning strategy exists if, and only if, a monotone winning strategy
exists.

Moreover, from any monotone winning strategy, a tree projection of $\HG_1$ w.r.t.~$\HG_2$ can be computed in polynomial time.
\end{theorem}

\begin{example}\label{ex:construction-bis}\em
Consider again the setting of Example~\ref{ex:construction}, and the strategy graph $G(\bar \sigma)$ shown in
Figure~\ref{fig:strategy-new-bis}. Note that the strategy $\bar \sigma$ is monotone, and in fact the moves of the Captain one-to-one correspond
to the hyperedges in the tree projection $\HG_a$. \hfill $\lhd$
\end{example}

The crucial properties to establish Theorem~\ref{thm:costruzione} are next recalled, as they will be useful in our subsequent analysis too. Let
$\sigma$ be a strategy, and let $v_p=(h_p,M_p,C_p)$ and $v_r=(h_r,M_r,C_r)$ be two configurations in its domain such that
$\sigma(v_p)=(h_r,M_r)$ and $C_r$ is a \ecomponent{v_p,M_{r}}. Let $\sigma(v_r)=(h_s,M_s)$ and define $\E((M_r,C_r),M_s)=M_r\cap \F(C_r)
\setminus M_{s}$ (which is equivalent to $\partial C_r\setminus M_{s}$ because $C_r$ is an \component{M_r}) as the \emph{escape-door} of the
Robber in $v_r$ when attacked with $M_s$. From~\cite{GS08}, a move is monotone if, and only if, such an escape door is empty; in particular,
$\sigma(v_r)$ is non-monotone if (and only if) $\E((M_r,C_r),M_s)\neq\emptyset$.
Let $M_r'=M_r\setminus \E((M_r,C_r),M_s)$, let $C_r'$ be the \component{M_r'} with $C_r\cup \E((M_r,C_r),M_s)\subseteq C_r'$, which exists
since $\E((M_r,C_r),M_s)\subseteq \F(C_r)$ and $M_r'\subseteq M_r$, and let $v_r'=(h_r,M_r',C_r')$. Finally, consider the following strategy
$\sigma'$:
\begin{equation}
\label{def:transformation}
\sigma'(h,M,C)= \left\{
\begin{array}{ll}

(h_r,M_r') & \mbox{ if }(h,M,C)=(h_p,M_p,C_p)\\
\sigma(h,M,C) & \mbox{ otherwise.}\\
\end{array}
\right.
\end{equation}

For such a state of the game, a number of technical properties have been proved in~\cite{GS08}. We summarize them in the following lemma.

\begin{lemma}[\cite{GS08}]\label{lem:construction}
The following properties hold:

\begin{enumerate}
\item[\em (1)] $\E((M'_r,C'_r),M_s)=\emptyset$.
\item[\em (2)] For each \iecomponent{v_p,M_r} $C$, either $C\subseteq C_r'$ or $C$ is a \iecomponent{v_p,M_r'}.
\item[\em (3)] For each \iecomponent{v_p,M_r'} $C'\neq C_r'$, $C'$ is a \iecomponent{v_p,M_r}.
\item[\em (4)] A set $C$ is a \iecomponent{v_r,M_s} if, and only if, it is a \iecomponent{v_r',M_s}.
\item[\em (5)] If $\sigma$ is a winning strategy, then $\sigma'$ is a winning strategy too.
\end{enumerate}
\end{lemma}

\subsection{Greedy Strategies}

Since winning strategies correspond to tree projections, there is no efficient algorithm for their computation. Indeed, just recall that
deciding the existence of a tree projection is not feasible in polynomial time, unless $\Pol=\NP$~\cite{GMS07}.

Our goal is then to focus on certain ``greedy'' strategies that are easy to compute. Intuitively, in greedy strategies it is required that
\emph{all} cops available at the current squad $h_p$ and reachable by the Robber enter in action. If all of them are in action, then a new
squad $h_r$ is selected, again requiring that all the active cops, i.e., those in the frontier, enter in action.

\begin{definition}\label{def:greedy}\em
On the Robber and Captain game played on $(\HG_1,\HG_2)$, a strategy $\sigma$ is {\em greedy} if, for any configuration $v_p=(h_p,M_p,C_p)$ in
the domain of $\sigma$, the next position $\sigma(v_p)=(h_{r},M_{r})$ is such that $M_{r}=h_{r}\cap \F(C_p)$, where $h_{r}=h_p$ if $h_p\cap
C_p\neq\emptyset$, and $h_{r}$ is any squad in $\edges(\HG_2)$ if $h_p\cap C_p=\emptyset$. 
\end{definition}

Given such a greedy way to select cops at each step, observe that the former case ($h_p\cap C_p\neq\emptyset$) may only occur if the Robber is
able to come back to some position previously controlled by the Captain. Greedy winning strategies are indeed non-monotone in general, and for
some pair of hypergraphs it is possible that there is no monotone winning greedy strategy, although  monotone winning strategies (non-greedy)
exist.

\begin{example}\label{ex:greedy-strat}\em
Consider again the hypergraphs $\HG_1$ and $\HG_2$ shown in Figure~\ref{fig:greedy-bis}, and recall that the strategy graph of a monotone
winning strategy $\bar \sigma$ is depicted in Figure~\ref{fig:strategy-new-bis}. However, there is no monotone greedy strategy in this case.
Indeed, if at the beginning of the game the Captain asks the squad $\{A,C,D,E,G\}$ to enter in action and the Robber goes on $B$, then in the
next move the Robber is forced to lose the control on $A$ in order to move on $\{C,B\}$ and eventually win via $\{B,A\}$---see again
Figure~\ref{fig:greedy-bis}. On the other hand, if the attack of the Captain starts on either side, say on the left branch, the Captain has
then to attack the component that includes the triangle and the other branch. At this point, the only available greedy choice is use the big
squad and hence to employ cops $\{C,D,E,G\}$. However, as in the previous case, $G$ will be later (necessarily) left free to the Robber, in
order to win the game. \hfill $\lhd$
\end{example}

We now show that, differently from arbitrary strategies, the existence of greedy winning strategies can be decided in polynomial time. To
establish the result, a useful technical property is that greedy strategies can only involve a polynomial number of configurations. Let us
denote by $\texttt{MaxGreedyStrat}(\HG_1,\HG_2)$ the maximum domain cardinality over any greedy strategy in the Robber and Captain game on a
pair $(\HG_1,\HG_2)$.

\begin{lemma}\label{lem:maxconf}
Let $(\HG_1,\HG_2)$ be a pair of hypergraphs. Then, $\texttt{MaxGreedyStrat}(\HG_1,\HG_2)$ is at most $|\edges(\HG_2)|\times|\nodes(\HG_1)|
(|\edges(\HG_2)|\times|\nodes(\HG_1)|+1) +1$.
\end{lemma}

\begin{proof}
Let $\sigma$ be a greedy strategy, and let $v_p=(h_p,M_p,C_p)$ be a configuration in its domain. Note that the only configuration where
$h_p=M_p=\emptyset$ is the starting configuration $(\emptyset,\emptyset,\nodes(\HG_1))$, which is taken into account by the final ``$+1$'' in
the statement. Therefore, we next assume $M_p\neq\emptyset$.
In order to establish the result, we shall derive an upper bound on the number of possible distinct configurations $(h_r,M_{r},C_{r})$ following
$(h_p,M_p,C_p)$ in the game where the Captain plays according to $\sigma$. In particular, we shall distinguish two cases, based on whether
$h_p\cap C_p$ is empty or not. For each of these two scenarios, we shall bound the number of such configurations $(h_r,M_{r},C_{r})$ according
to the constraints imposed on them by the definition of greedy strategy (cf. Definition~\ref{def:greedy}).

Consider the case where $h_p\cap C_p=\emptyset$. In this case, a new squad $h_{r}\in\edges(\HG_2)$ is chosen by the Captain according to
$\sigma$. Since $C_p$ is an \component{M_p} and thus $\partial C_p\subseteq M_p\subseteq h_p$, we get that this case occurs only if $C_p$ is
actually an \component{h_p}, too. Such a component is uniquely identified by any pair of the form $(h_p,X_p)$ such that $X_p\in \nodes(\HG_1)$
is a representative of the component (e.g., the node in $C_p$ having the smallest position according to any fixed ordering over the nodes).
It follows that the new set of cops $M_{r}=h_{r}\cap \F(C_i)$ is uniquely determined by $h_{r}$ and $C_p$ and thus may be identified through a
triple $(h_{r},h_p,X_p)$. Thus, the maximum number of such sets $M_{r}$ of cops is $|\edges(\HG_2)|^2\times |\nodes(\HG_1)|$. Moreover, the
possible configurations $(h_{r},M_{r},C_{r})$ following $(h_p,M_p,C_p)$ in the game where the Captain plays according to $\sigma$ are
identified by quadruples of the form $(h_{r},h_p,X_p,X_{r})$, where $h_{r}$ is used both to identify itself and to determine the set $M_{r}$
together with $h_p$ and $X_p$, and where $X_{r}$ is a representative of the \component{M_{r}}.
In fact, if there is no \ecomponent{v_p,M_p}, then $X_{r}$ is a distinguished element not in $\nodes(\HG_1)$ (or some element in $M_p$ occupied
by some cop) meaning that the only configuration following $(h_p,M_p,C_p)$ is $(h_{r},M_{r},\emptyset)$ where the Robber is captured.
Overall, the maximum number of such configurations is $|\edges(\HG_2)|^2 \times |\nodes(\HG_1)|^2$.

Finally, consider the case where $h_p\cap C_p\neq\emptyset$. In this case, $M_{r}=h_p\cap \F(C_p)$. Since $C_p$ is an \component{M_p},
$\partial C_p\subseteq M_p\subseteq h_p$. It follows that the new nodes from $\F(C_p)$ to be included in $M_{r}$ belong to $C_p$, that is, we
may also write $M_{r}=M_p\cup (h_p\cap C_p)$. Note that no configuration of the game following this one can be of this type. Indeed, every
\component{M_{r}} $C_{r}$ where the Robber may go from $C_p$ will be a subset of $C_p$ (because $\partial C_p\subseteq M_p\subseteq
M_{r}\subseteq h_p$), and will have intersections with $h_p$. As a further consequence, such a $C_{r}$ must be an \component{h_p}. By
contradiction, if there is some node $X_p\in C_{r}\subseteq C_p$ that is \connected{M_{r}} to some $X_{r}$ in $h_p\setminus M_{r}$, then $X_p$
is also \connected{M_p} to $X_{r}$. However, this is impossible because $X_p$ is also in $C_p$ and hence $X_{r}$ would be in $C_p$, too, and
hence in $h_p\cap C_p$ and in $M_{r}$, by construction. Therefore, the possible configurations $(h_p,M_{r},C_{r})$ following $(h_p,M_p,C_p)$ in
the game where the Captain plays according to $\sigma$ are identified by pairs of the form $(h_p,X_p)$, where $X_p\in \nodes(\HG_1)$ is the
representative of the \component{h_p} $C_{r}$ (and where $M_{r}$ is computed from them).
As above, if there is no \ecomponent{v_p,M_p}, then $X_p$ is a distinguished element witnessing that the configuration is a capture
configuration of the form $(h_p,M_{r},\emptyset)$.
Overall, the maximum number of such configurations is $|\edges(\HG_2)|\times |\nodes(\HG_1)|$.~\hfill~$\Box$
\end{proof}

\begin{figure}[t]
\vspace{2mm}
\centering
\fbox{\parbox{0.9\textwidth}{\small
\vspace{-1mm}\begin{itemize}
\item[]\hspace{-3mm}\emph{Boolean function} \textsc{GreedyWinningStrategy}$(h_p,M_p,C_p,i)$;
\vspace{0mm}\item[] 
/$\ast$  $(h_p,M_p,C_p)$ is an extended configuration over $(\HG_1,\HG_2)$,\\ \ \ \ \ \ $\hspace{10mm} i\geq 0$ is a natural number $\ast$/
\
\\ \vspace{-2mm}\hrule\ \\

\vspace{-4mm}
\item[1)] \textbf{if} $i>\texttt{MaxGreedyStrat}(\HG_1,\HG_2)$, then \textbf{return} \textsc{False};
\item[2)] \textbf{if} $h_p \cap C_p \neq \emptyset$, then  \textbf{let} $h_{r}=h_p$;
\item[] \textbf{else}  \textbf{guess} a hyperedge $h_{r}\in\edges(\HG_2)$;
\item[3)] \textbf{let} $M_{r}= h_{r} \cap \F(C_p) $;
\item[4)] \textbf{for each} \ecomponent{(h_p,M_p,C_p),M_{r}} $C_{r}$ \textbf{do}
\item[] \quad \ \textbf{if} \emph{not} \textsc{GreedyWinningStrategy}$(h_{r},M_{r},C_{r},i+1)$, then \textbf{return} \textsc{False};
\item[5)] \textbf{return} \textsc{True};

\vspace{-1mm}\end{itemize}}} \caption{\small \textsc{GreedyWinningStrategy}.}\label{fig:Algoritmo1}
\end{figure}

To see that the existence of a winning greedy strategy is decidable in polynomial time, consider the \textsc{GreedyWinningStrategy} algorithm
illustrated in Figure~\ref{fig:Algoritmo1}, which receives as input a configuration $(h_p,M_p,C_p)$ for the Robber and Captain game, plus a
``level'' $i$. Note that this algorithm is a high-level specification of an alternating Turing machine, say $\mathcal{M}_G$~\cite{john-90}.

After the first step, where we check that the number of recursive calls has not exceeded the number of all distinct configurations, the
algorithm suddenly evidences its \emph{non-deterministic} nature. Indeed, it guesses a hyperedge $h_{r}$ corresponding to the next move of the
Captain (existential step of $\mathcal{M}_G$). Eventually, it returns \textsc{True} if, and only if, the recursive calls
\textsc{GreedyWinningStrategy}$(h_{r},M_{r},C_{r},i+1)$ with $M_{r}= h_{r} \cap \F(C_p)$ succeed on each \ecomponent{(h_p,M_p,C_p),M_{r}}
(universal step of $\mathcal{M}_G$).

\begin{theorem}\label{thm:GreedyExistence}
Deciding the existence of a greedy winning strategy in the Robber and Captain game is feasible in polynomial time.
\end{theorem}

\begin{proof}
Let $(\HG_1,\HG_2)$ be a pair of hypergraphs, and consider the execution of the Boolean function $\textsc{GreedyWinningStrategy}$ on input the
starting $(\emptyset, \emptyset, \nodes(\HG_1),0)$.
In general, the function receives as its input a quadruple $(h_p,M_p,C_p,i)$, where $(h_p,M_p,C_p)$ is a configuration for a greedy strategy and
where $i$ counts the number of recursive invocations, i.e., the recursion level of the current invocation.
For the moment, let us get rid of step (1). Then, each invocation of the algorithm, based on the current configuration $(h_p,M_p,C_p)$,
computes the next position $(h_r,M_r)$ of the Captain. In particular, by Definition~\ref{def:greedy}, the position is completely determined by
the hyperedge $h_r$, which is ``guessed'' by the algorithm at step (2)---hence, the algorithm is non-deterministic. Given $h_r$, step (3) is
responsible for computing $M_r$ according to Definition~\ref{def:greedy}. Then, for all options $C_r$ that are available to the Robber after the
move $(h_r,M_r)$, the algorithm checks recursively at step (4) whether there is a winning strategy for the Captain. The algorithm returns
\textsc{True} if, and only if, all such recursive calls succeed.
We next show that the algorithm is correct and that it can be implemented to take polynomial time.

Concerning the correctness, due to its non-deterministic nature, it is easily seen that, by getting rid of step {(1)}, it returns \textsc{True}
if, and only if, the Captain has a greedy winning strategy in the game played on $(\HG_1,\HG_2)$ (which we assume to be ``visible'' by the
function at a every call, to avoid a longer signature).
Moreover, we claim that the check performed at step~{(1)} cannot lead to a wrong \textsc{False} output. Indeed, just observe that the number of
recursive calls is bounded by the number of all distinct configurations, which is $\texttt{MaxGreedyStrat}(\HG_1,\HG_2)$ at most, by
Lemma~\ref{lem:maxconf}. Therefore, if the recursion level $i$ exceeds this threshold, then we can safely answer \textsc{False}.

Let us now focus on the running time. We know that $\textsc{GreedyWinningStrategy}$ can be implemented on an alternating Turing machine
$\mathcal{M}_G$, whose existential steps correspond to the guess statements at step~2, while universal steps are used for checking that the
conditions at step~4 are satisfied by all the relevant components. In addition, by indexing the various data structures and by referring each
component via one point contained in it (selected through any fixed criterium), the machine can be implemented to use logarithmic many bits on
its worktape.

For instance, recall from the proof of Lemma~\ref{lem:maxconf} that every configuration is identified by at most four elements of the form
$(h_p,h_{r},X_p,X_{r})$ with $h_p,h_{r}\in\edges(\HG_2)$ and $X_p,X_{r}\in\nodes(\HG_1)$. Therefore, any configuration may be encoded by (at
most) four indexes whose maximum size is $\log \max \{|\edges(\HG_2)|,|\nodes(\HG_1)|\}$. Moreover, the check at step~{(1)} ensures that the
length of each branch of the computation tree of $\mathcal{M}_G$ is finite, and actually bounded by a polynomial in the size of the input.
For the sake of completeness, observe that all subtasks in the function, such as computing connected components and the like, are easily
implementable in nondeterministic logspace, so that such tasks just correspond to further (polynomially-bounded) branches of the computation
tree of $\mathcal{M}_G$. Thus, $\textsc{GreedyWinningStrategy}$ may be implemented as a \emph{log-space} alternating Turing machine, which
immediately entails the result, because Alternating Logspace is equal to Polynomial Time~\cite{chan-etal-81}.~\hfill~$\Box$
\end{proof}

It is well known that an alternating Turing machine $\mathcal{M}_G$ can be simulated by a standard machine in polynomial time. First, compute
the polynomially-many possible instant descriptions (IDs) of the machine, and build a graph representing the possible connections between any
pair of IDs, according to its transition relation. Then, evaluate this graph along some topological ordering as follows. Mark all IDs without
outcoming arcs associated with final accepting states; then mark all IDs associated with  existential states having a marked successor, or
associated with universal states, and whose successors are all marked. Then, the machine $\mathcal{M}_G$ accepts its input if, and only if, the
starting ID is marked. Moreover, the subgraph induced by the marked nodes encodes its accepting computations.

Moreover, from such a marked graph it is straightforward to compute the strategy graph of a greedy winning strategy, because IDs associated
with (children of) existential states encode the possible choices of the Captain.\footnote{For the sake of completeness note that, by using
these ideas, one might also provide  a direct dynamic programming algorithm to compute a strategy graph by using a bipartite graph representing
all possible configurations and positions of the Robber and Captain game. However, we find the non-deterministic function
$\textsc{GreedyWinningStrategy}$ more elegant and easy to present.} Just visit the graph starting from the initial configuration, but for each
ID associated with an existential state, select one child to be visited arbitrarily (all choices are marked and hence accepting).

\begin{corollary}\label{cor:computeGreedyStrat}
The strategy graph of a greedy winning strategy (if any) in the Robber and Captain game is computable in polynomial time.
\end{corollary}

\section{Larger Islands of Tractability}\label{sec:larger}

From the previous section (see Theorem~\ref{thm:costruzione} and Example~\ref{ex:greedy-strat}), we know that monotone winning strategies for
the Captain in the game over $(\HG_1,\HG_2)$ are associated with tree projections of $\HG_1$ w.r.t.~$\HG_2$, and that in some cases it is
possible that there is no monotone winning greedy strategy, although  monotone winning strategies (non-greedy) exist. In this section, we show
that from {\em any} (possibly non-monotone) greedy winning strategy a tree projection can be still computed in polynomial time. The key fact
here is that any non-monotone greedy strategy can be converted into a monotone one, though not a greedy one in general. Based on this
observation, a larger island of tractability for tree projections will be eventually singled out.

\subsection{Nice Strategies and Greedy Tree Projections}

To establish our results, it is useful to consider a special form of strategies that we call {\em nice} (for they remind the notion of nice
tree decompositions of graphs), where at every configuration the Captain first removes those cops that are no longer in the frontier.

Formally, $\sigma$ is a {\em nice} strategy if $\sigma(h_p,M_p,C_p)= (h_p,\partial C_p)$, whenever $\partial C_p \subset M_p$. Because such
inactive cops play no role in the Robber and Captain game, a winning nice strategy exists if (and only if) there exists a winning strategy, and
the same holds for greedy strategies. Just note that restricting the cops to the border of $C_p$ is a legal choice in greedy strategies (it
corresponds to the selection of the same squad $h_{r}=h_p$ before attacking the robber in the component $C_p$ with some further squad).
Clearly enough, such a nice strategy can be computed in polynomial time from any given strategy. Also, if desired, the polynomial time
algorithm for computing a greedy strategy may be easily adapted to compute directly a winning nice greedy strategy (if any).

\begin{figure}[t]
\centering
\includegraphics[width=0.8\textwidth]{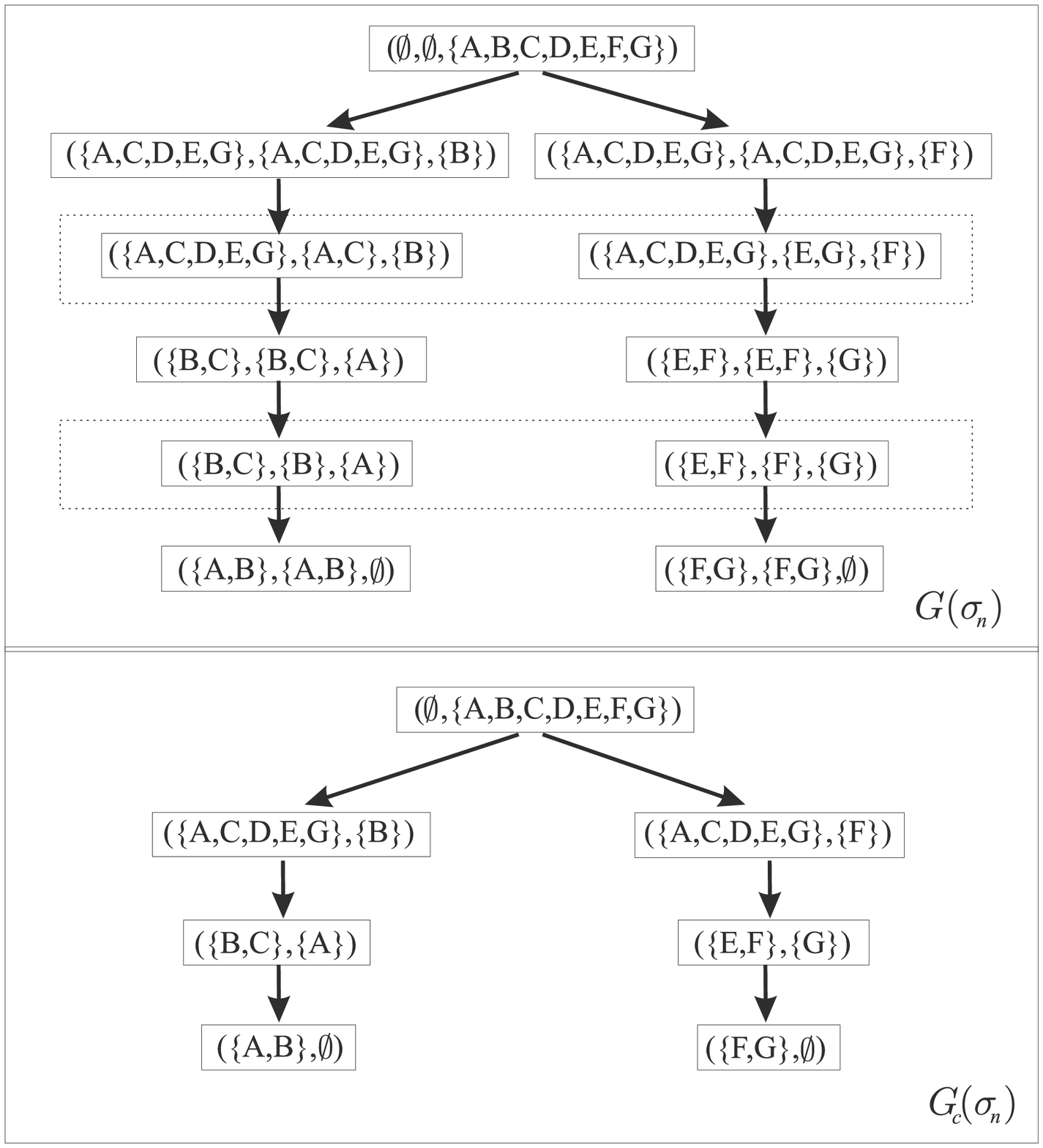}
\caption{The strategy and component graphs for the nice strategy $\sigma_n$ in Example~\ref{ex:nice}.}\label{fig:nice}
\end{figure}

\begin{example}\label{ex:nice}\em
Consider again the setting discussed in Example~\ref{ex:construction} and illustrated in Figure~\ref{fig:greedy-bis}. Note that the strategy
$\sigma$ is not nice. Indeed, Figure~\ref{fig:nice} reports the strategy graph associated with a strategy $\sigma_n$ that is nice and that is
obtained from $\sigma$ by just explicitly adding the configurations where the Captain has to remove the cops that are no longer in the
frontier. \hfill $\lhd$
\end{example}

The reason for introducing nice strategies is that they admit a compact representation. First, given any configuration $(h_p,M_p,C_p)$ and a
Captain's choice $M_{r}$, the \ecomponents{(h_p,M_p,C_p),M_{r}} for the Robber are determined by $C_p$ and $M_{r}$ only, because $\partial C_p$
is computable from $C_p$. Therefore, we use hereafter the simplified notation \ecomponent{C_p,M_{r}} to refer to this set of
\components{M_{r}}.
Moreover, in place of the strategy graph, we can use a \emph{component graph}, defined as follows.

\begin{definition}\em
Let $(\HG_1,\HG_2)$ be a pair of hypergraphs.
Let $G=(N,A)$ be a directed graph whose nodes are pairs of the form $(h_p,C_p)$, where $h_p\in\edges(\HG_2)$, and $C_p$ is either the emptyset
or a \component{\partial C_p} of $\HG_1$ such that $\partial C_p \subseteq h_p$. Then, we say that $G$ is a {\em component graph} if it meets
the following conditions:
\begin{enumerate}
\item[ (1)] There is a \emph{root} node $(\emptyset,\nodes(\HG_1))\in N$ that is the only node without incoming arcs.

\item[ (2)] Each node $(h_p,C_p)\in N$, with $C_p\neq \emptyset$, has outgoing arcs to $m\geq 0$ nodes $(h_{r},\bar C_1),\dots ,(h_{r},\bar
    C_m)$ such that, if $M_{r}$ is the set $\bigcup_{j=1}^m \partial \bar C_j \cup (C_p\setminus \bigcup_{j=1}^m \bar C_j)$, it holds that
    $M_{r}\subseteq h_{r}$ and the \ecomponents{C_p,M_{r}} are the components $\bar C_1,...,\bar C_m$.

\item[ (3)] Each node $(h_p,C_p)\in N$ has an outgoing arc to $(h_{r},\emptyset)$ if $C_p\subseteq h_{r}$. 
\end{enumerate}
\end{definition}

Note that every nice strategy $\sigma$ is encoded by the component graph $G_c(\sigma)=(N,A)$ defined as follows. There is a node $(h_p,C_p)$
(resp., $(h_p,\emptyset)$) in $N$ if there is a configuration $(h_p,\partial C_p,C_p)$ in the domain of $\sigma$ (resp., a capture
configuration $(h_p,\partial C_p,\emptyset$) induced by $\sigma$).
There is an arc in $A$ from a node $(h_p,C_p)$ to a node $(h_r,C_r)$ if there is an arc from $(h_p,M_p,C_p)$ to $(h_r,M_r,C_r)$ in the strategy
graph $G(\sigma)$. No more nodes and arcs occur in $N$ and $A$, respectively.
For instance, the graph depicted on the bottom part of Figure~\ref{fig:greedy-bis} is the component graph associated with the nice strategy
$\sigma_n$ of Example~\ref{ex:nice}.

Conversely, any component graph $G$ encodes a nice strategy $\sigma_{G}$, via the following procedure. Associate the root
$(\emptyset,\nodes(\HG_1))$ with the initial configuration $(\emptyset,\emptyset,\nodes(\HG_1))$. Inductively, assume that a node $(h_p,C_p)$
is associated with a configuration $(h_p,M_p,C_p)$, and that $(h_{r},\bar C_1),...,(h_{r},\bar C_m)$ are the labels of the nodes having an
incoming arc from $(h_p,M_p,C_p)$. Let $M_{r}= \bigcup_{j=1}^m \partial \bar C_j \cup (C_p\setminus \bigcup_{j=1}^m \bar C_j)$, with
$M_{r}\subseteq h_{r}$. Then, define $\sigma_{G}(h_p,M_p,C_p)=(h_{r},M_{r})$, and define $\sigma_G(h_{r},M_{r},\bar C_j)= (h_{r},\partial \bar
C_j)$, with $j\in \{1,...,m\}$, in the case where $\partial \bar C_j \subset M_{r}$.

\begin{theorem}\label{thm:greedy}
A tree projection of $\HG_1$ w.r.t.~$\HG_2$ can be computed in polynomial time if the Captain has a greedy winning strategy on $(\HG_1,\HG_2)$.
\end{theorem}
\begin{proof}
By Theorem~\ref{thm:GreedyExistence}, we can decide in polynomial time whether a winning greedy strategy for the Captain in the game played on
$(\HG_1,\HG_2)$ exists or not. In the negative case, we are done. Otherwise, we compute in polynomial time a winning nice greedy strategy
$\sigma$ (or turn a given strategy into a nice one), and we subsequently build its component graph $G_c(\sigma)$.
Let us now make a copy $G'=(N',A')$ of $G_c(\sigma)$, and note that $G'$ is a directed acyclic graph, because it encodes a winning strategy.

The line of the proof is to show that we can incrementally modify $G'$, until we end up with a component graph still encoding a nice winning
strategy that is however monotone; the result then follows because a tree projection
can be computed in polynomial time from a monotone winning strategy~\cite{GS08}.
We process $G'$ from the leaves to the root, according to any of its topological orderings. Whenever we process a node $v_j$ of $G'$ that is associated with a non-monotone move, we shall eventually apply a local transformation to $G'$, discussed in the
steps (i)--(iv) detailed below and illustrated in Example~\ref{ex:illustra}. In particular, we shall show that the updated graph is still a
component graph encoding a nice winning strategy. Moreover, we shall observe that, after the execution of steps (i)--(iv), the game encoded in
the modified graph and starting at $v_j$ is monotone. Hence, after the root is processed, we end up with a monotone winning strategy. The proof
will be completed by observing that the number of such transformations is polynomially bounded.

Let us now formalize the approach sketched above. Let $\overrightarrow N =v_1,\dots,v_{|N'|}$ be the topologically ordered sequence of the
nodes of $G'$, where the nodes without outgoing arcs, called leaves, are in the first positions, and the node without incoming arcs, its root,
is at the last position. Note that leaves correspond to capture configurations for the robber, while the root $v_{|N'|}=(\emptyset,
\nodes(\HG_1))$ is associated with the starting configuration $(\emptyset, \emptyset, \nodes(\HG_1))$ of the game. Moreover, if $(v,v')\in A'$,
the node $v$ is said to be a parent of $v'$, while $v'$ is said to be a child of $v$. Then, we start modifying the graph $G'$, by navigating
the sequence $\overrightarrow N$ using an index $j$, as discussed below.

Starting with $j=1$,  while $j< |N'|$, consider the current node $v_j$ in the sequence, associated with a configuration $(h_j,M_j,C_j)$
(initially, the first leaf) in the domain of $\sigma_{G'}$. If every child of $v_j$ is labeled by some $(h'',C'')$ with $C''\subseteq C_j$,
then let index $j := j+1$ and continue the ``while'' loop, or stop and output the current graph $G'$ if $v_j$ is the root. Otherwise, let $v_s$
be a child of $v_j$ labeled by $(h_s,C_s)\in N'$ such that $C_s\not\subseteq C_j$, and associated with the configuration $(h_s,M_s,C_s)$. That
is, $\sigma_{G'}(h_j,M_j,C_j)=(h_s,M_s)$ is a non-monotone move. Then, take any parent $v_p$ of $v_j$, and let $(h_p,M_p,C_p)$ the
configuration associated with $v_p$ (whose label is thus $(h_p,C_p)$). Modify the graph so that  $\sigma_{G'}(h_p,M_p,C_p)=(h_j,M_j')$, where
$M_j'=M_j\setminus \E(v_j,M_s)$. In particular, let $C'_j$ be the \component{M'_j} that properly includes $C_j$, and for which thus $C_p\cup
C'_j$ is \connected{M_p\cap M'_j}. Then, the modified component graph will also encode the choice $\sigma_{G'}(h_j,M_j',C_j')=(h_j,\partial
C_j')$ if $\partial C_j' \subset M_j'$,  and $\sigma_{G'}(h_j,\partial C_j',C_j')=(h_s,M_s)$. The transformation of the graph is as follows:

\begin{itemize}
\item[(i)] Add a node $v_j'$ labeled by $(h_j,C'_j)$ to $N'$ and to the sequence $\overrightarrow N$ in the position before $v_j$, and add
    to $A'$ an arc from $v_j'$  to each child of $v_j$, i.e., to nodes labeled by $(h_s,C'')$, for each  \ecomponent{C_j',M_s} $C''$.

\item[(ii)] Remove from $A'$ all outgoing arcs of $v_p$ to nodes whose labels do not contain \ecomponents{C_p,M_j'} (in particular, the arc
    towards $v_j$ is removed).

\item[(iii)] Add to $A'$ an arc from $v_p$ to $v_j'$.

\item[(iv)] Remove from $N'$ any node different from the root which is left without incoming arcs, and continue the ``while'' loop
    considering again node $v_j$, or the next available node in $\overrightarrow N$ if $v_j$ has been removed by $N'$.
\end{itemize}

\begin{figure}[t]
\centering
\includegraphics[width=0.99\textwidth]{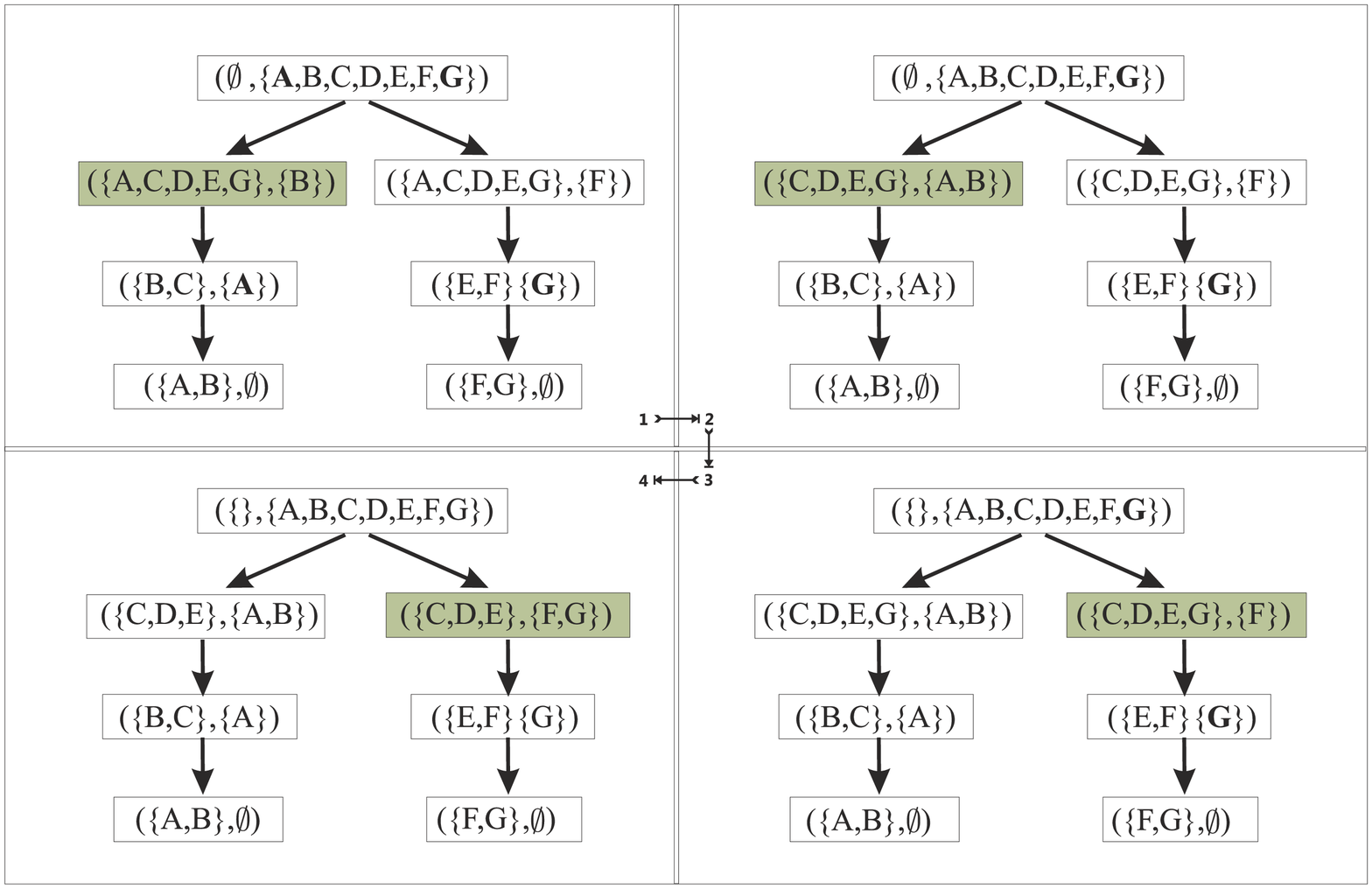}
\caption{Illustration of the algorithm in the proof of Theorem~\ref{thm:greedy}.}\label{fig:greedybis}
\end{figure}

\begin{example}\label{ex:illustra}\em
The application of the above procedure to the nice strategy $\sigma_n$ discussed in Example~\ref{ex:nice} is illustrated in
Figure~\ref{fig:greedybis}. Note that two non-monotone moves are removed in total. Note that, at the end of the transformation, we get a
component graph encoding precisely the monotone strategy $\bar \sigma$, whose strategy graph is illustrated in
Figure~\ref{fig:strategy-new-bis}. \hfill $\lhd$
\end{example}

First observe that every iteration of the loop at step~1 above, precisely implements on the graph $G'$ the transformation (of the non-monotone
strategy encoded by $G'$) described by Expression~\eqref{def:transformation}, and whose properties are described by
Lemma~\ref{lem:construction}. In more detail, with these properties in mind, by executing steps~(i)--(iii) we replace the Captain's choice
$(h_j,M_j)$ at $(h_p,M_p,C_p)$ by the new choice $(h_j,M_j')$, and we get the following situation: (a) Because of the new choice $M_j'$, only
one new $\ecomponent{C_p,M_j'}$ is available to the robber, that is, the $\component{M_j'}$ $C_j'$ properly including the $\component{M_j}$
$C_j$. As a consequence, at step~(i) the one node $v_j'$ corresponding to this component is added to $N'$.  (b) The $\ecomponents{C_j',M_s}$
are the same as the $\ecomponents{C_j,M_s}$, so that  the outgoing arcs of $v_j'$ will be the same as the node $v_j$. That is, we keep the same
winning strategy as before, as the Robber's options after the Captain's choice $M_s$  are the same as before (and hence the Captain knows how
to successfully attack them). (c) The set of $\ecomponents{C_p,M_j'}$, with the exception of the new $C_j'$, are a subset of the
$\ecomponents{C_p,M_j}$. In fact, some components may collapse after the new choice of the Captain. Then, at step~(iv), we remove the nodes
associated with $\ecomponents{C_p,M_j}$ that are now left without incoming arcs. For instance, it is possible that we delete $v_j$ if $v_p$ was
its only parent, or it is possible that we delete some nodes associated with collapsed components. Note that the new graph $G'$ obtained from
these steps is still a component graph, hence it encodes a (new) nice strategy $\sigma_{G'}$.

Therefore, Lemma~\ref{lem:construction} entails that, after each iteration and thus after the entire procedure, the strategy $\sigma_{G'}$ is a
winning strategy. We claim that it is actually a monotone winning strategy, by a simple inductive argument: if $v_j$ is the current node, after
the execution of steps (i)--(iv), $\sigma_{G'}$ is a monotone winning strategy for the game starting at the configuration $v_j$. Then, the
claim follows because, for $j=|N'|$, it means that  $\sigma_{G'}$ is a monotone winning strategy for the whole game starting at the root.

The base case is when the algorithm starts at $j=1$, and hence the statement holds because the first position in $\overrightarrow N$ is
occupied by some leaf, which is a capture configuration of the winning strategy.
Now assume that the statement holds for $j-1$, and consider the execution of the above procedure on node $v_j$. Note that the proposed
transformation deals with just one (possibly new) component $C_j'$ instead of the strictly smaller $C_j$; everything else in the strategy does
not change, in particular no node preceding $v_j$ in the topological order is affected by the transformation.
Then, the monotonicity of the strategy on the game starting at $v_j$ immediately follows from the induction hypothesis and from
Lemma~\ref{lem:construction}.(1), which says that $\E(v_j',M_s)=\emptyset$ and hence that this move is monotone, so that $C''\subseteq C'_j$,
for each $\ecomponent{C_j',M_{i+1}}$ $C''$.

Because each iteration in feasible in polynomial time, it just remains to show that the whole procedure requires at most polynomially many
iterations. To this end, note that whenever some node $v_j$ encodes a non-monotone move, one node $v_j'$ is added to $N'$ for each parent $v_p$
of $v_j$. Indeed, the node $v_j$ is considered again after the first iteration where it was evaluated, if it still has incoming arcs (see
step~(iv)). However, after steps (i)--(iv), $\sigma_{G'}$ is a monotone winning strategy for the game starting at the new configuration $v_j'$.
Therefore, no new node will be subject to further transformations in subsequent iterations along the given topological ordering of $N'$. It
follows that the number of iterations of the described procedure is bounded by $\nodes(G_c(\sigma))\times {\it MaxIn}$, where ${\it MaxIn}$ is
the largest in-degree over the nodes of $G_c(\sigma)$. Thus, the number of iterations is bounded by a polynomial in the size of the strategy
graph of the greedy winning strategy, which is in its turn polynomial in the size of $(\HG_1,\HG_2)$.

Finally, from the monotone winning strategy  $\sigma_{G'}$ encoded by the output $G'$ of the above procedure, a tree projection $\HG_a$ of
$(\HG_1,\HG_2)$ is available. Just define $\nodes(\HG_a)=\nodes(\HG_1)$ and $\edges(\HG_a)= \{ M \mid \sigma_{G'}(v)=(h,M)$ {\em for some
configuration}  $v$ {\em in the domain of } $\sigma_{G'}\}$. See~\cite{GS08}, for more detail about such a relationship between monotone
strategies and tree projections.~\hfill~$\Box$
\end{proof}

With the above result in place, let $\C_{gtp}$ denote the class of all pairs $(Q,\V)$ such that there exists a greedy winning strategy $\sigma$
for the Captain in the Robber and Captain game on $(\HG_Q,\HG_\V)$. As shown in the proof of Theorem~\ref{thm:greedy}, based on $\sigma$ a tree
projection of $\HG_Q$ w.r.t.~$\HG_\V$, which we call \emph{greedy tree projection}, can be computed in polynomial time. Therefore, the
following is established.

\begin{corollary}\label{cor:greedyTP}
$\C_{gtp}$ is an island of tractability.
\end{corollary}

\subsection{Captain vs Marshal}\label{sec:capVSmar}

A class of tractable pairs related to our class $\C_{gtp}$ has been defined in~\cite{adler08} in terms of the \emph{Robber and Marshal game}
played by one Marshal and the Robber on the hypergraphs $(\HG_1,\HG_2)$. This game has been originally defined on a single hypergraph to
characterize hypertree decompositions~\cite{gott-etal-03}, and its natural extension to pairs of hypergraphs has been defined and studied
in~\cite{adler08}.

The game is as follows.
The Marshal may control one hyperedge of $\HG_2$, at each step.
The Robber stands on a node and can run at great speed along hyperedges of $\HG_1$; however, (s)he is not permitted to run through a node that
is controlled by the Marshal. Thus, a \emph{configuration} is a pair $(h,C)$, where $h$ is the hyperedge controlled by the Marshal, and $C$ is
an \component{h} where the Robber stands.
Let $(h_p,C_p)$ be a configuration. This is a capture configuration, where the Marshal wins, if $C_p\subseteq h_p$. Otherwise, the Marshal
moves to another hyperedge $h_{r}\in \edges(\HG_2)$; while (s)he moves, the Robber may run through those nodes that are left by the Marshal or
not yet occupied. Thus, the Robber selects an \component{h_{r}} $C_{r}$ such that $C_{r}\cup C_p$ is \connected{h_p\cap h_{r}}.
We say that the Marshal has a \emph{winning strategy} if, starting from the initial configuration $(\emptyset,\node)$, (s)he may end up the
game in a capture position, no matter of the Robber's moves. A winning strategy is \emph{monotone} if the Marshal may monotonically shrink the
set of nodes where the Robber stands.

Because only nodes in the frontier are actually used at each step in the monotone Robber and Marshal game, this game and the monotone variants
of the Robber and Captain game clearly define the same hypergraph properties.

\begin{fact}\label{prop:greedy-monotone}
The following are equivalent:
\begin{itemize}
\item[(1)] There is a monotone winning strategy for the Marshal in the Robber and Marshal game on $(\HG_1,\HG_2)$.

\item[(2)] There is a monotone winning greedy-strategy for the Captain in the Robber and Captain game on $(\HG_1,\HG_2)$.
\end{itemize}
\end{fact}

Let $\C_{rm}$ denote the class of all pairs $(Q,\V)$ such that there exists a  monotone winning strategy for the Marshal on $(\HG_Q,\HG_\V)$.
From the results in~\cite{adler08,adler-thesis}, $\C_{rm}$ is an island of tractability as well. However, the set of tractable instances
identified by greedy winning strategies in the Robber and Captain game properly includes this class. The reason is that greedy winning
strategies are allowed to be non-monotone.

\begin{theorem}\label{thm:comparazione}
$\C_{rm} \subset \C_{gtp}$.
\end{theorem}
\begin{proof}
Because greedy strategies are not required to be monotone, $\C_{rm} \subseteq \C_{gtp}$ follows from Fact~\ref{prop:greedy-monotone}.
For the proper inclusion, just consider again Example~\ref{ex:greedy-strat}. The pair of hypergraphs shown in Figure~\ref{fig:greedy-bis} is
such that the Marshal has no monotone winning strategy, while the Captain has a (non-monotone) winning greedy strategy.\footnote{This example
is in fact inspired by a similar simpler pair of hypergraphs where no monotone strategy for the Marshal exists, described
in~\cite{adler08}.}~\hfill~$\Box$
\end{proof}

For completeness, recall that the non-monotone variant of the Marshal and Robber game is instead too powerful to be useful. Indeed, there are
pairs of hypergraphs where the Marshal has a non-monotone winning strategy but no tree projection exists. We refer the interested reader
to~\cite{adler08} for more detail about the monotonicity gap in the Robber and Marshal game.

\section{Applications}\label{sec:methods}

In this section, we explore two applications of the results derived about greedy tree projections. In particular, we first move from the
general setting of tree projections to analyze specific decomposition methods, and we then focus on tree projections for queries to be answered
over databases whose relations have ``small'' arities.

\subsection{Greedy Hypertree Decompositions and Further Greedy Methods}

The tractability result about the general case of greedy tree projections can be immediately applied to every structural decomposition method,
in order to get new tractable variants of these methods.
To carry out the elaborations, observe that any structural decomposition method {\tt DM} can be viewed as a method associating a set $\V$ of
views to any given query $Q$. Indeed, the decompositions of $Q$ according to {\tt DM} are precisely tree projections of $\HG_Q$
w.r.t.~$\HG_\V$.

Given this correspondence, it is then natural to consider the greedy variant of any structural decomposition method {\tt DM}, denoted by
$\textit{greedy-}{\tt DM}$, whose associated decompositions are the \emph{greedy} tree projections of $\HG_Q$ w.r.t.~$\HG_\V$.
From Corollary~\ref{cor:greedyTP}, every decomposition method, possibly an intractable one such as the generalized hypertree decomposition
method, defines an island of tractability by means of its greedy variant.

\begin{fact}
Let {\tt DM} be a structural decomposition method and let $\textit{greedy-}{\tt DM}$ be its greedy variant. Then, the class of all queries
having a $\textit{greedy-}{\tt DM}$ decomposition is recognizable in polynomial time, and every query in the class may be evaluated in
polynomial time over any given database.
\end{fact}

We next focus on the greedy variant of the method based on generalized hypertree decompositions. Let $k\geq 1$. Recall from
Section~\ref{sec:framework} that the width-$k$ generalized hypertree decompositions of a query $Q$ are the tree projections of
$(\HG_Q,\HG_Q^{k})$. Indeed, we are considering one distinct view over each set of variables that can be covered by at most $k$ query-atoms.

\begin{definition}\em
A width-$k$ {\em greedy hypertree-decomposition} (we omit ``generalized'', for short) of a conjunctive query $Q$ is any greedy tree projection
of $(\HG_Q,\HG_Q^{k})$. Accordingly, the {\em greedy (generalized) hypertree-width} of $Q$, denoted by $\textit{gr-hw}$, is the smallest $k$
such that $Q$ has a greedy hypertree decomposition of width $k$.
\end{definition}

This greedy variant provides a new tractable approximation of the (intractable) notion of generalized hypertree decomposition, which is better
than (standard) hypertree decompositions.

\begin{figure}[t]
\centering
\includegraphics[width=0.98\textwidth]{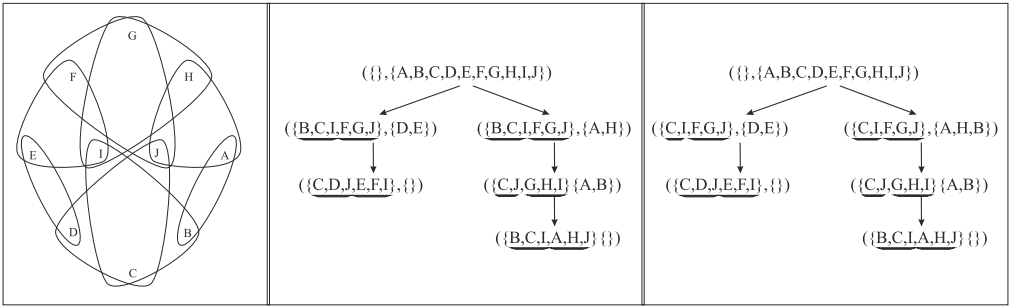}
\caption{Examples in the proof of Fact~\ref{greedy-hw}.}
\label{fig:GMS}
\end{figure}

\begin{fact}\label{greedy-hw}
For any query $Q$, $\textit{ghw}(Q)\leq \textit{gr-hw}(Q) \leq \textit{hw}(Q)$ holds. Moreover, there are queries $Q$ for which
$\textit{gr-hw}(Q) < \textit{hw}(Q)$, even for $\textit{gr-hw}(Q)=2$.
\end{fact}
\begin{proof}
The first relationship is immediate: in the first inequality we use the fact that greedy hypertree decompositions are a special case of
generalized hypertree decompositions, while the second inequality holds because the notion of hypertree decomposition is characterized by the
monotone Robber and Marshals game, played on $\HG_Q$ by a Robber and $k$ Marshals~\cite{gott-etal-03}. This game is equivalent to play the
monotone game with one Marshal on the pair of hypergraphs $(\HG_Q,\HG_Q^{k})$, which is the same as playing the monotone Robber and Captain
game.

For the strict upper bound  $\textit{gr-hw}(Q) < \textit{hw}(Q)$, consider the query $Q_0$, taken from~\cite{CJG08,GMS07}, whose hypergraph
$\HG_{Q_0}$  is depicted in the left part of Figure~\ref{fig:GMS}. For this query, it is shown in~\cite{GMS07} that $\textit{hw}(Q_0)=3$ and
$\textit{ghw}(Q_0)=2$. However, $\textit{gr-hw}(Q_0)=2$ holds. Indeed, there is a winning greedy strategy for the Captain in the game played on
$(\HG_{Q_0},\HG_{Q_0}^{2})$, as shown in the central part of Figure~\ref{fig:GMS}, and thus there exists a greedy tree projection of
$\HG_{Q_0}$ w.r.t.~$\HG_{Q_0}^{2}$.

In the figure, the set of selected cops at each step is underlined in such a way that the reader may identify the original pair of hyperedges
from $\HG_{Q_0}$ that forms the chosen squad in $\HG_{Q_0}^{2}$. Note that the strategy is non-monotone, as it is witnessed by the right branch
where the Robber can return on the node $B$. However, by using the construction in Theorem~\ref{thm:greedy}, it can be turned into a monotone
(while not greedy) one, by removing the escape door $B$ in the first move of the Captain (see the right part of the figure). From the monotone
strategy, we immediately get the desired tree projection.~\hfill~$\Box$
\end{proof}

More general examples are given by the {\em subedge-based} decomposition methods, defined in~\cite{GMS07}. Recall that a subedge-method ${\tt
DM}$ is based on a function $f$ associating with each integer $k \geq 1$ and each hypergraph $\HG_Q=(V,E)$ of some query $Q$ a set $f (\HG_Q,
k)$ of subedges of $\HG_Q$, that is, a set of subsets of hyperedges in $E$. Moreover, the set of width-$k$ ${\tt DM}$-decompositions of $Q$ can
be obtained as follows: (1) obtain a hypertree decomposition $\HD$ of $\HG_f = (V, E\cup f (\HG, k))$, and (2) convert $\HD$ into a generalized
hypertree decomposition of $\HG_Q$ by replacing each subedge $h\in f (\HG_Q, k)\setminus E$ occurring in $\HD$ by some hyperedge $h'\in E$ such
that $h\subseteq h'$ (which exists because $h$ is a subedge).

Because such a method is based on width-$k$ hypertree decompositions, in the tree projection framework it can be recast as follows. A width-$k$
${\tt DM}$-decomposition is any tree decomposition of $\HG_Q$ w.r.t.~$\HG_f^k$ associated with some monotone winning strategy of the Robber and
Marshal game on this pair of hypergraphs. On the other hand, according to its greedy variant $\textit{greedy-}{\tt DM}$, the width-$k$
decompositions are the greedy tree projections of $\HG_Q$ w.r.t.~$\HG_f^k$. It follows that the greedy variant of this method is more powerful.

\begin{fact}\label{fact:subedge}
Let ${\tt DM}$ be any {subedge-based} decomposition method. Let $k\geq 1$ and let $Q$ be a query. Then, a width-$k$ ${\tt DM}$-decomposition of
$Q$ exists only if a width-$k$ $\textit{greedy-}{\tt DM}$-decomposition of $Q$ exists. The converse does not hold, in general.
\end{fact}
\begin{proof}
The first entailment follows from Theorem~\ref{thm:comparazione}. The fact that the converse does not hold in general, follows from
Fact~\ref{greedy-hw}, because the hypertree decomposition method is a subedge-based method (based on the function
$f(\HG_Q,k)=\emptyset$).~\hfill~$\Box$
\end{proof}

This is a remarkable result, as in~\cite{GMS07} some examples of subedge-based decomposition methods, such as the {\em component hypertree
decompositions}, are shown to generalize most previous proposals of tractable structural decomposition methods, such as hypertree and
spread-cut decompositions (in fact, all of them, but the approximation of fractional hypertree decomposition, later introduced in~\cite{M09}).
From Fact~\ref{fact:subedge}, their greedy variants are even more powerful.

\subsection{Tractability over Small Arity Structures}\label{smallArities}

We now consider the case of relational structures having small arity, which is a relevant special case in real-world applications.
In fact, observe that any variable that is not involved in any join operation in a conjunctive query (that is, any variable that occurs in one
atom only) is irrelevant and may be projected out in a preprocessing phase. It follows that the {\em effective arity} to be considered in our
structural techniques is actually determined by the largest number of variables that any atom has in common with other atoms (i.e., those
variables involved in join operations), independently of the arity of the relations in the original database schema. This number is often
small, in practice.\footnote{In fact, it is easy to further generalize this line of reasoning, by considering as ``effective arity'' the
maximum cardinality over the hyperedges in the GYO-reduct of $\HG_Q$. (Recall that the GYO reduct  of a hypergraph is obtained by iteratively
removing nodes that occur in one hyperedge only and hyperedges included in other hyperedges, until no further removal is possible---see,
e.g.,~\cite{ullm-89}.)}

Therefore, it is interesting to investigate whether the general problem of computing a tree projection of a pair of hypergraphs is any easier
in the case of small arity structures (for the sake of presentation, we just consider here the standard structure arity, leaving to the
interested reader the straightforward extension to the above mentioned ``effective arity''). We next show that the problem is indeed in
polynomial-time for bounded-arity structures, and it is moreover {\em fixed-parameter tractable (FPT)}, if the arity is used as a parameter of
the problem. This is not difficult to prove, but it was never stated before (as far as we know), and we believe it is important to pinpoint
this tractability result.

Recall that a problem is FPT if there is an algorithm that solves the problem in {\em fixed-parameter polynomial-time}, that is, with a cost
$f(k) O(n^{O(1)})$, for some computable function $f$ that is applied to the parameter $k$ only. In other words, this algorithm not only runs in
polynomial time if $k$ is bounded by a fixed number, but it also exhibits a ``nice'' dependency on the parameter, because $k$ is not in the
exponent of the input size $n$. Let \textit{p-TP} be the problem of computing a tree projection of $\HG_Q$ w.r.t.~$\HG_{\V}$, for a given pair
$(Q,\V)$, parameterized by the maximum arity of the relations occurring in $(Q,\V)$.

\begin{theorem}\label{theo:fp}
The problem $\textit{p-TP}$ is fixed-parameter tractable.
\end{theorem}
\begin{proof}
Let $(Q,\V)$ be an input pair for $\textit{p-TP}$, let $(\HG_Q,\HG_{\V})$ be the pair of associated hypergraphs, and let $k$ be the parameter.

Compute the simplicial version $\HG_s$ of the hypergraph $\HG_{\V}$, that is, the hypergraph having the same set of nodes as $\HG_{\V}$, and
where $\edges(\HG_s)= \{ h'\neq\emptyset \mid h'\subseteq h, h\in \edges(\HG_{\V})\}$. Therefore, $\edges(\HG_s)$ contains all subsets of every
hyperedge of $\HG_{\V}$.
Clearly, $\HG_s$ can be computed in time $O(2^k \times |\edges(\HG_{\V})|)$, and the tree projections of $(\HG_Q,\HG_{\V})$ are the same as the
tree projections of $(\HG_Q,\HG_s)$. To conclude, observe that any tree projection of the latter pair can be computed in polynomial-time by
Theorem~\ref{thm:greedy} and the fact that, having a squad for every possible set of cops in any squad/hyperedge of $\HG_{\V}$, the greedy
strategies in the game Robber and Captain on $(\HG_Q,\HG_s)$ are precisely the (unrestricted) strategies in the Robber and Captain game on
$(\HG_Q,\HG_{\V})$, which characterize the tree projections of $(\HG_Q,\HG_{\V})$.\footnote{Note that the same relationship holds for the
monotone strategies and, hence, for the Marshal's strategies in the Robber and Marshal game over the pair $(\HG_Q,\HG_s)$, as observed by
Adler~\cite{adler-thesis}.}~\hfill~$\Box$
\end{proof}

The above tractability result is smoothly inherited by all structural decomposition methods {\tt DM} such that the arity of the views is
$O(f(k))$ for some computable function $f$ that does not depend on the size of the input. For instance, this is the case for the methods based
on bounded (generalized hyper)tree decompositions, but not for fractional hypertree decompositions. In particular, if $w$ is the fixed maximum
width for a class of queries having bounded generalized hypertree width,  the maximum arity of the computed views is $w\times k$. Thus, if
$\textit{p-ghw}_w$ denotes the problem of computing a width-$w$ generalized hypertree decomposition of a query, parameterized by the maximum
arity of the query atoms, we immediately get the following result.

\begin{corollary}\label{fpt-ghd}
The problem $\textit{p-ghw}_w$ is fixed-parameter tractable.
\end{corollary}

\section{From Theory to Practice}\label{sec:practice}

Many recent works are using structural methods based on the computation of a tree projection of the given instance, such as generalized
hypertree decompositions or fractional hypertree decompositions, for answering queries to relational databases or solving constraint
satisfaction problems (CSPs), where constraints are represented as finite relations encoding the allowed tuples of values. Moreover, structural
methods find applications in game theory and combinatorial auctions, as well as in other fields (see~\cite{GGLS16} for more information and
references on these applications, with a focus on hypertree decompositions). This is quite natural because we are actually using a basic
hypergraph-theoretic notion that, in principle, may be useful in any application where acyclic instances are easy to deal with. In the rest of
the section, we discuss some of these applications.

\subsection{Using Tree Projections}

Tree projections represent transformations from a given problem to its acyclic variant. For instance, consider a conjunctive query $Q$ over a
database instance $\DB$, and assume that its associated hypergraph $\HG_Q$ is cyclic. Given a set of available views $\V$, any tree projection
$\HG_a$ of $\HG_Q$ with respect to $\HG_{\V}$ can be used to obtain an acyclic query $Q'$ on a database $\DB'$ that is  equivalent to $Q$ on
$\DB$:\footnote{We actually assume that the relations associated with views are not more restrictive than the original query. This is always
the case for the mentioned structural method. For a formal treatment of the general case, see~\cite{GS10b}.} For each hyperedge  $h$ of
$\HG_a$, compute a fresh atom such that its set of variables is $h$ and its relation is obtained by projecting on $h$ the relation associated
with any view $w\in\V$ whose set of nodes includes $h$ (such a view exists by definition of tree projection).
This immediately provides a polynomial-time upper bound on the running time of answering the query. Let $r$ be the size of the largest relation
associated with the views in $\V$, and let $m$ be the number of hyperedges of $\HG_a$, which is known to be bounded by the number of variables
(in the so-called normal form tree projection~\cite{GS08}). The above transformation is feasible in $O(m \cdot r)$, with each relation in the
new database $\DB'$ having at most $r$ tuples. Let $r'\leq r$ be the actual size of the largest relation of the database $\DB'$. The overall
complexity immediately follows by adding the cost of evaluating the new acyclic instance (e.g., by Yannakakis's algorithm~\cite{yann-81}),
which dominates the overall cost: the worst-case upper bound is $O(m \cdot (r'+s) \cdot \log (r'+s))$ time and $O(m \cdot (r'+s))$ space, where
$s$ is the size of the output.

As a further example of applications of methods based ont tree projections, we mention the EmptyHeaded relational engine that uses generalized
hypertree decompositions in its query planner~\cite{aberger-etal-16}. In~\cite{KalinskyEK16}, similar techniques based on structural
decompositions have been used to guide a flexible caching of intermediate results in the context of computing multiway joins.
In~\cite{amroun-etal-16}, a CSP solving technique based on generalized hypertree decompositions and using compressed representations for the
relations has been proposed, and its scalability has been assessed over well-known CPS benchmarks.

Finally note that algorithms based on such structural methods can be parallelized, as pointed out in~\cite{gott-etal-01}. Generalized hypertree
decompositions are indeed used for parallel query answering in the GYM algorithm~\cite{AfratiJRSU14}, which is a distributed and generalized
version of Yannakakis' algorithm for answering acyclic queries, specifically designed for the MapReduce framework~\cite{DeanG10}.

\subsection{Views beyond Structural Decomposition Methods}

Consider a pair of hypergraphs $\HG_1 \leq \HG_2$, of which we want to compute a tree projection $\HG_a$, with  $\HG_1 \leq \HG_a\leq \HG_2$.
The resource hypergraph $\HG_2$,  whose hyperedges define what we have called views, is completely arbitrary in the general tree projection
framework we deal with. As we have seen, specific algorithms for defining views lead to different decomposition methods. We mentioned methods
where views are computed in polynomial time (when we talk about islands of tractability), but the tree projection framework is actually much
more general.

Views may represent any subproblem that we can use to solve the given instance, or that is already available from previous computations (e.g.,
materialized views in databases). In some applications, one may relax the polynomial-time constraint and consider instead fixed-parameter
tractable computations, for some (application-specific) parameter. In other applications, views may be associated with subproblems that can be
solved by using non-structural properties. With this respect, we mention an important line of research in constraint satisfaction, looking at
restrictions on the form of specific (fixed) constraint relations, regardless of the structure of constraint scopes, see, e.g.,~\cite{CJ06}.

There are also hybrid approaches, looking at both structure and data~\cite{GS14,CZ11}. In concrete applications on big databases, the hybrid approach is
mandatory: views should be subqueries such that their computation cost is estimated to be low (that is, less than some given threshold),
according to information on the actual database instance, such as selectivity of attributes, keys, cardinality of relations, indices, and so
on. In~\cite{Greco2007,Ghionna2011}, a query optimizer taking into account a simple cost model for subquery evaluation, together with views
based on the hypertree decomposition method, has been implemented. The optimizer can be put on top of any existing database management system
supporting JDBC technology, by transparently replacing its standard optimization module. The results demonstrate a significant gain obtained by
using query plans based on hypertree decompositions on queries involving more than two atoms. Further implementations directly inside
open-source Database Management Systems are subjects of current work.

\subsection{In Practice}

There is room for practical applications of the results presented in this paper to improve the efficiency of the above solutions, besides the
theoretical interest in providing a better understanding of the  difference between the power of general strategies and the power of controlled
non-monotonicity in the Robber and Captain game on pair of hypergraphs.

Consider the result on the fixed-parameter tractability of computing a tree projection of $\HG_1$ w.r.t. $\HG_2$, where the maximum cardinality
of the hyperedges in $\HG_2$ (that is, the arity of views, in database terms) is used as the parameter, say $k$, of the problem. We proved that
 a tree projection, if any, can be computed in $O(2^k n^{O(1)})$.
We believe that this is a useful result because it means that, in all those instances where the number of variables is not large, an effective
query optimization (with respect to arbitrary views and hence with respect to any decomposition method) is feasible in reasonable time. Indeed,
the computation of the decomposition depends only on hypergraphs (and not on the database) and, unlike other fixed-parameter algorithms, the
algorithm described in Theorem~\ref{theo:fp} is ``practical,'' as there are no (hidden) huge constants and the dependence on the arity
parameter is single-exponential. This is of particular interest in database applications, where small queries over huge amount of data are the
typical instances. Furthermore, in such a context, the same queries are frequently run over a varying database, so that a good query
optimization pays over the time.

We can be even more concrete by focusing on the specific decomposition method based on generalized hypertree decompositions. By
Corollary~\ref{fpt-ghd}, a width-$w$ generalized hypertree decomposition of the hypergraph $\HG_Q$ of a given query $Q$ can be computed in
$O(2^k n^{O(1)})$, where $k$ is the maximum number of variables occurring in any query atom. We next point out that it is very convenient to
look for decompositions with the smallest possible widths, which means using the most powerful decomposition methods (that are affordable in
the available optimization time). Say $n$ be the combined size of the query $Q$ and the database, and consider the query answering problem
parameterized by the generalized hypertree width of $Q$, say $w$. It is well known that this problem is not fixed-parameter tractable, which
means that (under usual fixed-parameter complexity assumptions) an exponential dependency on the parameter of the form $O(n^{f(w)})$ is
unavoidable. It follows that even small savings in the width leads to exponential savings in the query evaluation time (and here $n$ includes
the database size). It is worthwhile noting that the same exponential dependency holds if we consider as parameter $w$ the notion of width
associated with other mentioned decomposition methods, in particular the treewidth. We thus argue that investing some time in computing
low-width decompositions is very convenient even for queries having small arities.
Indeed, $\textit{ghw}(Q)\leq \textit{tw}(Q)$ always holds,
and for some queries $\textit{tw}(Q) = k \cdot \textit{ghw}(Q)$.

The main algorithmic result of this paper, that is, the notion of greedy tree projection and its tractability, is particularly interesting
whenever we deal with instances having large hypergraphs. This is often the case in constraint satisfaction problems, where there are instances
with hundreds of constraints, for which the computation of a generalized hypertree decomposition having the minimum possible width may not be
affordable. Many practical approaches for these applications adapt heuristics developed for the tree decomposition method, or use
the notion of hypertree width (see, e.g.,~\cite{dech-03,amroun-etal-16}). However, as pointed out above, if we are able to find better decompositions, we are guaranteed an
exponential-saving in the (worst-case) computation time. In this respect, using greedy tree projections may be a good choice. In particular,
the greedy method that we called greedy hypertree decomposition provides always better (or equal) results than hypertree decompositions (and hence than
tree decompositions), and it is computable in polynomial time for any fixed, bound on the width.

\section{Related Literature on ``Cops and Robbers'' Games}\label{sec:related}

In this paper we are mainly interested in games defined over hypergraphs or pairs of hypergraphs, such as those studied in~\cite{adler04} (see Section~\ref{sec:capVSmar}).
 We are not aware of many further works of this kind, apart from the
  {\em Robber and Army game}~\cite{GM06}, which was defined to approximate the notion of {fractional hypertree decomposition}. This game is indeed a variation of the Robber and Marshals game that characterizes hypertree decompositions, but this time marshals are replaced by a more powerful general, who is in charge of an army of $r$ battalions of soldiers (with $r$ being a rational number).
 The general may distribute her soldiers on the hyperedges in any arbitrary way (rational allocations are allowed).
  A node of the hypergraph is blocked (the robber cannot go through that node) if the number of soldiers on all hyperedges that contain this node adds up to the strength of at least one battalion. The game is then played in a monotonic way, like the Robber and Marshals game.

 As a matter of fact, all these games, comprising the {Robber and Captain} game~\cite{GS08} at the core of the present work, can be viewed as variations of the
  Robber and Cops game defined by Seymour and Thomas over
graphs~\cite{ST93}, in order to characterize the notion of treewidth.
  In this game, a number of Cops have to capture a Robber that can run at
great speed along the edges of a {graph}, while being not permitted to run trough a node that is controlled by a Cop. In particular, the Cops
can move over nodes by using helicopters and, before they land, the Robber is fast and can run trough those nodes that are left or not
yet occupied before the move is completed. A graph has treewidth bounded by $k$ if, and only if, there is a winning strategy for $k+1$ Cops in
this game~\cite{ST93}.
Unlike the Robber and Marshals (or Captain, or Army) game,
in the Robber and Cops game, restricting strategies to be monotone does not reduce in any way the power of cops.

By looking at the game defined by Seymour and Thomas, one might naturally wonder what happens if  the use of helicopters is not allowed, so that
Cops must move along the edges of the graph, precisely as the Robber does.
The study of this variant goes back to the eighties, when it was introduced by Winkler and Nowakowski~\cite{NW83} and independently by
Quilliot~\cite{Q83} under the name of the \emph{Cops and Robbers} game.
Since then, this game has been the subject of intense study (see the book by Bonato and Nowakowski~\cite{BN11}, and the references therein).
In particular, in the original formulation, the game proceeds in rounds, each consisting of a Cop turn followed by a Robber turn. In each
round, each cop may remain on her current vertex or move to an adjacent vertex, after which the robber likewise chooses to remain in place or
move to an adjacent vertex. For this game, several efforts have been spent to characterize the \emph{cop number} of (classes of) graphs, i.e.,
the minimum number of Cops needed to capture the Robber, regardless of her moves.

Graphs of cop number 1 were characterized already in the above mentioned seminal papers~\cite{NW83,Q83}. These graphs are based on a suitable
linear ordering of their vertices and can be recognized in polynomial time. Similarly, for any fixed natural number $k>0$, deciding whether a
graph has cop number bounded by $k$ is feasible in polynomial time; indeed, just notice that the number of possible different configurations
is $O(n^{k+1})$ for a graph with $n$ vertices. However, for $k$ being part of the input, it has been recently shown that it is $\rm
EXP$-complete to decide if the cop number of a graph does not exceed $k$~\cite{K15}, hence confirming a long-standing conjecture by Goldstein
and Reingold~\cite{GR95}---further complexity results for variants of the game can be found in the works by Fomin et al.~\cite{FGP12} and
Mamino~\cite{M13}.
For general graphs on $n$ vertices it is known that $\Omega(\sqrt{n})$ cops may be needed, and the celebrated Meyniel's conjecture~\cite{F87}
states that the cop number of a connected vertex graph is $O(\sqrt{n})$. Moreover, exact or approximate values of the cop number of several
classes of graphs have been derived so far, including plan graphs~\cite{AF84}, bounded genus graphs~\cite{S01}, and intersection
graphs~\cite{GJKK13}, just to name a few.

A variant of the Robber and Cops game discussed above assumes that the Robber is faster than the Cops in that, at each move, she can transverse
$s\geq 1$ edges of the graph. This variant has been introduced by Fomin et al.~\cite{FGKNS10}, who also showed that computing the cop number in
this setting is $\rm NP$-hard, for every fixed $s$, even on classes of split graphs. In particular, the case where the Robber has an unbounded
speed ($s=\infty$) is very related with the game by Seymour and Thomas characterizing treewidth, except for the use of helicopters. In fact, it turns out that, over planar graphs $G$, the cop number for this game is $\Theta(tw(G))$~\cite{AM15}. Further results
on this variant can be found in the work by Frieze et al.~\cite{FKL12}.
Yet another interesting variant of the Cops and Robber game where the robber is invisible has been also studied in the literature
(see~\cite{DDTY15}, and the references therein).

While the above games are defined and mainly studied over graphs,
the extension of graph games to hypergraph games is sometimes natural.
For instance, properties of the hypergraph version of the classical Cops and Robber game in~\cite{NW83,Q83} are studied in~\cite{Baird}.



\section{Conclusion and Future Work}\label{sec:conclusion}

By exploiting a recent game characterization of tree projections, we identified new islands of (structural) tractability based on a greedy
version of the powerful non-monotonic strategies in hypergraph games. We show that such greedy strategies can be computed efficiently, and can
always guide us towards the computation of useful decomposition trees. In fact, the proposed approach immediately provides larger ``greedy''
extensions of the most powerful structural decomposition methods defined in the literature.

Furthermore, again using the game-theoretic characterization of tree projections, we pinpoint the fixed-parameter tractability of this notion
(and hence of most structural decomposition methods) when the arity is used as the parameter. This models what happens if small arity instances
are considered, which often occurs in practice.

We believe that the results presented in the paper may be very useful in real-world applications and we are currently working on direct implementations of the proposed techniques in open-source database management systems.
Moreover, note that these results find applications in all those problems that can be solved efficiently on acyclic and  quasi-acyclic
instances, even outside the Database area we focused on. In particular, our results can be exploited immediately for solving Constraint
Satisfaction Problems.

Besides the implementation of efficient algorithms for the computations of greedy hypertree decompositions of large hypergraphs, a research
question regards the distance between generalized and greedy hypertree decompositions. From the known relationships with hypertree
decompositions, we immediately get that, for any hypergraph $\HG$, $\textit{ghw}(\HG) \leq \textit{gr-hw}(\HG) \leq \textit{hw}(\HG) \leq 3
\cdot \textit{ghw}(\HG) +1$. However, whether or not these bounds are tight is currently open.

Moreover, it would interesting to investigate whether the greedy techniques used in this paper for the Robber and Captain game on pair of
hypergraphs can be somehow useful for other kinds of (hyper)graphs games, such as those described in Section~\ref{sec:related}.

\end{document}